\documentclass[transmag]{IEEEtran}

\usepackage{subcaption}
\usepackage{amsmath,amsthm,amssymb,amsfonts}
\usepackage[hidelinks]{hyperref}
\usepackage{bm}
\newcommand{\RomanNumeralCaps}[1]{\MakeUppercase{\romannumeral #1}}
\theoremstyle{remark}

\newcommand{\bmA}{\mathbf A}
\newcommand{\bmI}{\mathbf I}

\newcommand{\bmB}{\mathbf B}
\newcommand{\bmC}{\mathbf C}
\newcommand{\bmD}{\mathbf D}
\newcommand{\bmP}{\mathbf P}
\newcommand{\bmG}{\mathbf G}

\newcommand{\bLambda}{\boldsymbol{\Lambda}}
\newcommand{\bPsi}{\boldsymbol{\Psi}}
\newcommand{\bmV}{\mathbf V}
\newcommand{\bmU}{\mathbf U}

\usepackage{xspace}
\usepackage{amssymb}
\usepackage{graphicx}
\usepackage{amsmath}
\usepackage{mathtools}
\theoremstyle{plain}
\newtheorem{thm}{Theorem}[] 
\newtheorem{lemma}{Lemma}

\usepackage{graphicx} 
\usepackage{amsmath, amsfonts, amssymb, amsbsy,nccmath} 
\usepackage{algorithm} 

\usepackage{enumerate} 
\usepackage{algorithmic} 
\usepackage{lipsum} 
\usepackage[sort,compress]{cite} 
\usepackage{epsfig} 
\usepackage{epstopdf}
\usepackage{mathtools}
\usepackage{dsfont} 
\usepackage[inline]{enumitem}
\usepackage{color,soul}

\usepackage{stfloats}  
\usepackage{tabularx}
\usepackage{xcolor}
\usepackage{lipsum}
\usepackage{mwe}
\usepackage{setspace}
\usepackage{longtable}
\usepackage{subcaption}
\usepackage[normalem]{ulem}
\usepackage[user]{zref}
\newcommand{\sot}[1]{} 

\newcounter{revc}
\makeatletter \zref@newprop{revcontent}{} \zref@addprop{main}{revcontent}
\zref@newprop{revsec}{} \zref@addprop{main}{revsec}
\zref@newprop{revpage}{} \zref@addprop{main}{revpage}

\newcommand{\revi}[2]{
\zref@setcurrent{revsec}{\thesection}%
\zref@setcurrent{revpage}{\thepage}%
\zref@setcurrent{revcontent}{#2}%
\refstepcounter{revc}%
\label{#1}%
\zlabel{#1}%
\textcolor{blue}{#2}%
}

\newcommand{\revinu}[2]{%
\zref@setcurrent{revsec}{\thesection}%
\zref@setcurrent{revcontent}{#2}%
\refstepcounter{revc}%
\zlabel{#1}%
\label{#1}
#2 }

\newcommand{\revr}[2]{%
\zref@setcurrent{revsec}{\thesection}%
\zref@setcurrent{revcontent}{#2}%
\refstepcounter{revc}%
\zlabel{#1}%
\label{#1} \sot{#2}} \makeatother

\expandafter\def\expandafter\quote\expandafter{\quote\onehalfspacing\fontsize{12}{14}\selectfont}

\usepackage[framemethod=tikz]{mdframed}
\usepackage{lipsum}

\definecolor{mycolor}{rgb}{0.122, 0.435, 0.698}

\newmdenv[innerlinewidth=0.5pt, roundcorner=4pt,linecolor=mycolor,innerleftmargin=6pt,
innerrightmargin=6pt,innertopmargin=6pt,innerbottommargin=6pt]{mybox}

\def\BibTeX{{\rm B\kern-.05em{\sc i\kern-.025em b}\kern-.08em T\kern-.1667em\lower.7ex\hbox{E}\kern-.125emX}}

\begin{document}
\sloppy
\title{Practical Hybrid Beamforming for Millimeter Wave Massive MIMO Full Duplex with Limited Dynamic Range}

\author{Chandan Kumar Sheemar, \IEEEmembership{Student Member, IEEE}, Christo Kurisummoottil Thomas, \IEEEmembership{Member, IEEE}, \\ and Dirk Slock, \IEEEmembership{Fellow, IEEE}

\thanks{Chandan Kumar Sheemar and Dirk Slock are with the Communication Systems Department at EURECOM, Sophia Antipolis, 06410, France (emails:sheemar@eurecom.fr,slock@eurecom.fr);  }
\thanks{Christo Kurisummoottil Thomas is with Qualcomm Finland RFFE Oy, Keilaranta 8, 02150 Espoo
(e-mail: ckurisum@qti.qualcomm.com).}
}

\IEEEtitleabstractindextext{\begin{abstract}
Full Duplex (FD) radio has emerged as a promising solution to increase the data rates by up to a factor of two via simultaneous transmission and reception in the same frequency band. This paper studies a novel hybrid beamforming (HYBF) design to maximize the weighted sum-rate (WSR) in a single-cell millimeter wave (mmWave) massive multiple-input-multiple-output (mMIMO) FD system. Motivated by practical considerations, we assume that the multi-antenna users and hybrid FD base station (BS) suffer from the limited dynamic range (LDR) noise
due to non-ideal hardware and an impairment aware HYBF approach is adopted by integrating the traditional LDR noise model in the mmWave band. In contrast to the conventional HYBF schemes, our design also considers the joint sum-power and the practical per-antenna power constraints. A novel interference, self-interference (SI) and LDR noise aware optimal power allocation scheme for the uplink (UL) users and FD BS is also presented to satisfy the joint constraints. The maximum achievable gain of a multi-user mmWave FD system over a fully digital half duplex (HD) system with different LDR noise levels and numbers of the radio-frequency (RF) chains is investigated. Simulation results show
that our design outperforms the HD system with only a few RF chains at any LDR noise level. The advantage of having amplitude control at the analog stage is also examined, and additional gain for the mmWave FD system becomes evident when the number of RF chains at the hybrid FD BS is small.
\end{abstract}

\begin{IEEEkeywords}
Millimeter wave, Full Duplex, Hybrid Beamforming, Limited Dynamic Range, Minorization-Maximization
\end{IEEEkeywords}
}

\maketitle
\section{Introduction} \label{Intro}
\IEEEPARstart{T}{he} revolution in wireless communications has led to an exponential increase in the data rate requirements and number of users. The millimeter wave (mmWave) frequency band $30 - 300$~GHz can accommodate the ever-increasing data demands and results to be a vital resource for future wireless communications \cite{pi2011introduction}. 
It offers much wider bandwidths than the traditional cellular networks, and the available spectrum at such higher frequencies is $200$ times greater \cite{rangan2014millimeter}. Full Duplex (FD) communication in mmWave has the potential to further double the spectral efficiency by offering simultaneous transmission and reception in the same frequency band. Moreover, it can be beneficial for efficient management of the vast mmWave spectrum, reducing end-to-end delays/latency, enabling advanced joint communication and sensing, and solving the hidden node problem \cite{liu2019hidden,parlin2020full,kabir2019scalable,barneto2021full}.

Self-interference (SI), which can be $90-110$ dB higher than the received signal \cite{gan2019full,rosson2019towards}, is a key challenge to achieve an ideal FD operation. Given the tremendous amount of SI, signal reception is impossible without a proper SI cancellation scheme. Beamforming is a powerful tool for FD to mitigate the SI while serving multiple users and can lead to a significant performance gain compared to a half duplex (HD) system \cite{huberman2014mimo,aquilina2017weighted,riihonen2012analog,day2012full,taghizadeh2016transmit,cirik2016beamforming,antonio2014sinr,biswas2020design,ciriklinear_rev,taghizadeh2018hardware}. However, its gain in practical communication systems is restricted by the limited dynamic range (LDR) of the radio-frequency (RF) chains \cite{day2012full}. The signal may suffer from LDR noise due to the distortions introduced by non-ideal power amplifiers (PAs), analog-to-digital-converters (ADCs), digital-to-analog-converters, mixers and low noise PAs. These impairments dictate the residual SI power which cannot be cancelled and therefore establish the achievable gain for FD \cite{day2012full}. This adverse effect urges the requirement of impairment aware beamforming designs and investigating their performance in terms of the LDR noise levels such that correct conclusions on the achievable gain of FD could be drawn. Such an approach for the fully digital FD systems can be adopted with the well-established LDR noise model available in \cite{aquilina2017weighted,riihonen2012analog,day2012full,taghizadeh2016transmit,cirik2016beamforming,antonio2014sinr,biswas2020design,ciriklinear_rev,taghizadeh2018hardware}. In general, impairment aware beamforming is more robust to distortions and can significantly outperform the naive schemes \cite{schenk2008rf,aghdam2019distortion}, see, e.g., \cite[Figure 2]{aghdam2019distortion}.

The deployment of multi-user mmWave FD systems requires the FD base stations (BSs) to be equipped with a massive number of antennas to overcome the propagation challenges. Owing to the hardware cost, they will have to rely on a hybrid architecture consisting of only a few RF chains. Therefore, efficient hybrid beamforming (HYBF) schemes are required for such transceivers to manage the SI and interference jointly by performing large-dimensional phasor processing in the analog
domain and lower-dimensional digital processing.

\subsection{State-of-the-art and Motivation}

In \cite{lopez2019analog,EURECOM+6506,satyanarayana2018hybrid,thomas2019multi,balti2020modified,huang2020learning,palacios2019hybrid}, novel HYBF designs for a point-to-point mmWave massive MIMO (mMIMO) FD system are studied. HYBF schemes of mMIMO FD relays and integrated access and backhaul are presented in \cite{abbas2016full,cai2020two,han2019full} and \cite{sheemar2021massive}, respectively. HYBF designs with single antenna uplink (UL) and downlink (DL) users for a single-cell and a multi-cell mmWave FD system are proposed in \cite{da20201} and \cite{zhao2020robust}, respectively. In \cite{roberts2020hybrid}, HYBF for mmWave mMIMO FD with only one UL and one DL multi-antenna user, under the receive LDR is proposed. In \cite{roberts2019beamforming}, HYBF for two fully connected mMIMO FD nodes that approaches SI-free sum-spectral efficiency is proposed. In \cite{roberts2020equipping}, HYBF for a mmWave FD system equipped with analog SI cancellation stage is presented. 
In \cite{sheemar2021hybridMIMO}, HYBF to generalize the point-to-point mmWave mMIMO FD communication to the case of a K-pair links is presented. Frequency-selective HYBF for a wide-band mmWave FD system is studied in \cite{roberts2020frequency}.

The literature on multi-antenna multi-user mmWave FD systems is limited only to the case of one UL and one DL user  \cite{roberts2020hybrid,roberts2019beamforming,roberts2020equipping,roberts2020frequency}. In \cite{roberts2020hybrid}, the receive side LDR of FD BS is also considered, which is dominated by the quantization noise of the ADCs. However, LDR noise from the transmit side is ignored, which also affects the performance of FD systems significantly \cite{korpi2014full}. The effect of cross-interference generated from the UL user towards the DL user is also not considered in \cite{roberts2020hybrid}, which can have a major impact on the achievable performance. Cross-interference generated from the neighbouring cells is well investigated in the dynamic time-division-duplexing networks \cite{sheemar2021game,da2021full,kim2020dynamic,de2019bidirectional,guo2018dynamic}, and it is more harmful to the multi-user FD systems as it occurs in the same cell. For example, consider the case of a small cell, in which BSs and users are expected to operate with a similar amount of transmit power \cite{guo2018dynamic}. Suppose that one FD BS simultaneously serves one UL and one DL user and that both the users are located close to each other and sufficiently far from the BS. In such a case, cross-interference can become as severe as the SI and can completely drown the useful signal intended for the DL user if not considered in the beamforming design. In a multi-user scenario with multiple UL users located near the DL users, each DL user suffers from cross-interference, which is summed over all the UL users' transmit power, with each UL user transmitting with a similar amount of power as the BS. In such a case, cross-interference can become even more severe than the SI if not considered in the design.

\subsection{Main Contributions} \label{contributo}
We present a novel HYBF design to maximize the weighted sum-rate (WSR) in a single-cell mmWave mMIMO FD system, i.e., for multiple multi-antenna UL and DL users. The users are assumed to have a limited number of antennas and digital processing capability. The FD BS is assumed to have a massive number of antennas and hybrid processing capability. Our design is based on alternating optimization and relies on the mathematical tools offered by minorization-maximization \cite{StoicaSelen:SPmag0104}. The users and BS are assumed to be suffering from the LDR noise due to non-ideal hardware, modelled with the traditional LDR model \cite{day2012full} and by extending it to the case of a hybrid transceiver, respectively.
Our work represents the first-ever impairment aware HYBF approach for mmWave FD and its analysis as a function of the LDR noise levels. Extension of the LDR noise model presented herein is applicable to any mmWave FD scenario.

In contrast to the conventional HYBF designs for mmWave FD, in this work, the beamformers are designed under the joint sum-power and the practical per-antenna power constraints. The sum-power constraint at each terminal is imposed by the regulations, which limits its total transmit power. In practice, each transmit antenna is equipped with its PA\footnote{The mMIMO systems are also expected to be deployed with one PA per-antenna to enable the deployment of very low-cost PAs \cite{larsson2014massive}.} \cite{yu2007transmitter} and the per-antenna power constraints arise due to power consumption limits imposed on the physical PAs \cite{lan2004input,chaluvadi2019optimal,yu2007transmitter,EURECOM+6498,cirik2016linear}. We also present a novel SI, interference, cross-interference and LDR noise aware optimal power allocation scheme to meet the joint constraints. 

Compared to the digital part, optimization of the analog stage is more challenging as it must obey the unit-modulus constraint. Recently, new transceivers have started to emerge, which with the aid of amplitude modulators (AMs), also allow amplitude control for the analog stage  \cite{castellanos2018hybrid,majidzadeh2018rate,roberts2020hybrid}. Such transceivers alleviate the unit-modulus constraint but require additional hardware.
Hence, we study both the unit-modulus and AMs cases and investigate when the amplitude control for mmWave FD could be advantageous.
In practice, as the analog beamformer and analog combiner can assume only finite values, a quantization constraint is also imposed on them during the optimization process. In our problem formulation, the WSR does not depend on the digital combiners, which are omitted in the design. They must be chosen as the minimum-mean-squared-error (MMSE) combiners after the convergence of the proposed algorithm. By omitting the digital combiners, equal to the sum of the number of UL and DL users, the HYBF design simplifies, and the per-iteration computational complexity reduces significantly.  

Simulation results show that our design outperforms a fully digital HD system and can deal with the SI, interference and cross-interference with only a few RF chains. Results are reported with different LDR noise levels, and significant performance gain is observed at any level.

In summary, the contributions of our work are:
\begin{itemize}

\item Extension of the LDR noise model for the mmWave band.

\item Introduction of the WSR maximization problem formulation for HYBF in a single-cell mmWave mMIMO FD system affected by the LDR noise.

\item A novel SI, interference, cross-interference, LDR noise and the practical per-antenna power constraints aware HYBF design.
 
\item Investigation of the achievable WSR in a multi-user mmWave FD system as a function of the LDR noise.

\item Optimal interference, SI, LDR noise and the per-antenna power constraints aware power allocation scheme for the hybrid FD BS and UL users.

\end{itemize}

\textbf{Paper Organization:} The rest of the paper is organized as follows. Section \ref{sytem_model} presents the system model, problem formulation and extends the LDR noise model. Sections \ref{simplificazione_problem} and \ref{sec_hybf} present the minorization-maximization method and a novel HYBF design, respectively. Finally, Sections \ref{simulazioni} and \ref{conclusioni} present the simulation results and conclusions, respectively.

\textbf{Mathematical Notations:} Boldface lower and upper case characters denote vectors and matrices, respectively. $\mathbb{E}\{\cdot\}, \mbox{Tr}\{\cdot\}, (\cdot)^H, (\cdot)^T$, $\otimes$, $\bm{I}$, $\bm{D}_{d}$ and $i$
denote expectation,
trace, conjugate transpose, transpose, Kronecker product, identity matrix, $d$ dominant vectors selection matrix and the imaginary unit, respectively. $\mbox{vec}(\bm{X})$ stacks the columns of $\bm{X}$ into a vector $\bm{x}$ and $\mbox{unvec}(\bm{x})$ reshapes $\bm{x}$ into $\bm{X}$. $\angle\bm{X}$ and $\angle\bm{x}$ return the unit-modulus phasors of $\bm{X}$ and the unit-modulus phasor of $\bm{x}$, respectively. $\mbox{Cov}(\cdot)$ and diag$(\cdot)$ denote the covariance and diagonal matrices, respectively. $\mbox{SVD}(\bm{X})$ returns the singular value decomposition (SVD) of $\bm{X}$. Element of $\bm{X}$ at the $m$-th row and $n$-th column is denoted as $\bm{X}(m,n)$. Vector of zeros of size $M$ is denoted as $\bm{0}_{M \times 1}$. Operators $|\mathbf{X}|$ and $|x|$ return a matrix of moduli of $\mathbf{X}$ and the modulus of scalar $x$, respectively.

 \begin{figure} 
	\centering
  \includegraphics[width=8.5cm, height=7cm ]{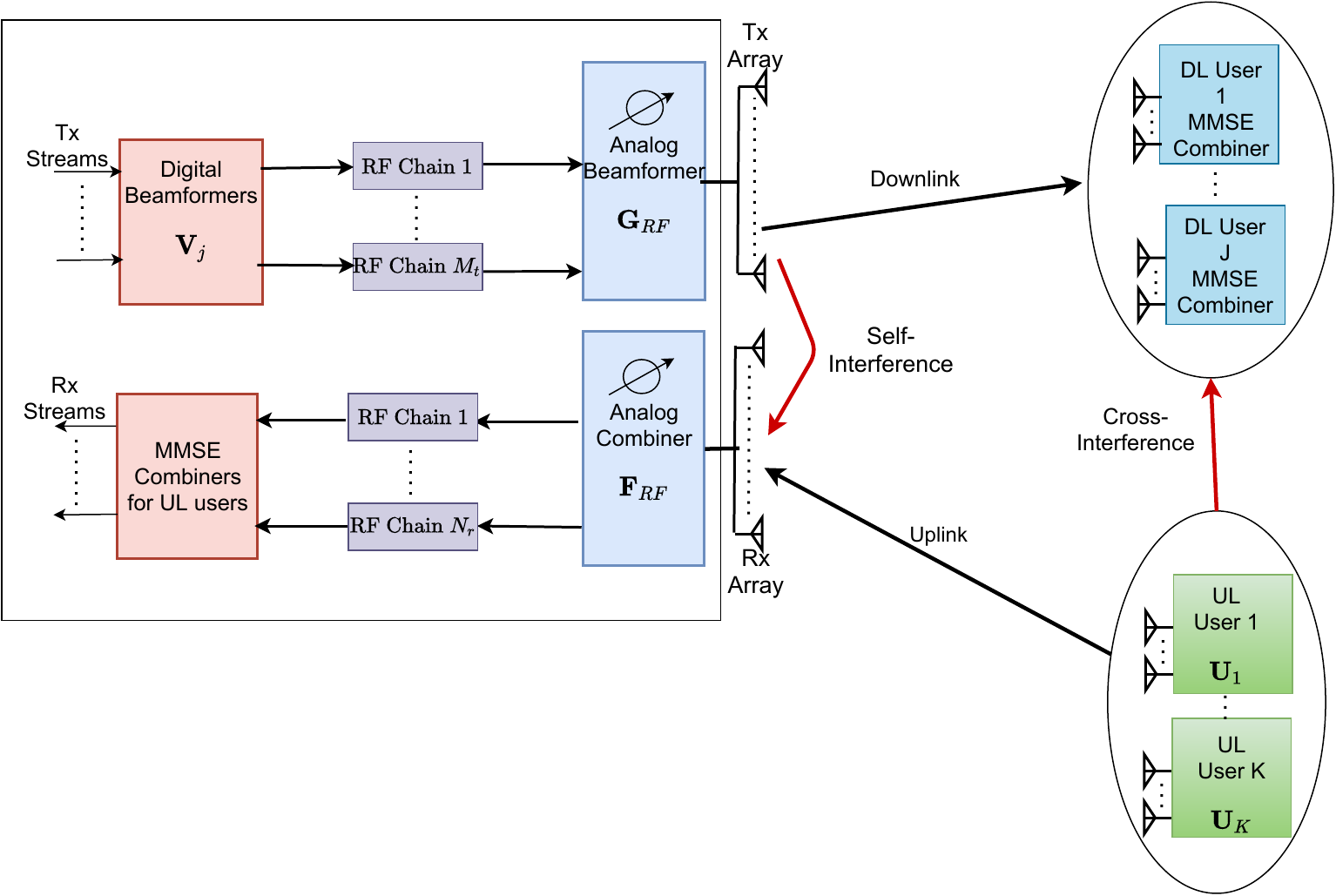}
  	\centering
 	\caption{FD in mmWave with HYBF to serve multi-antenna users. Tx and Rx denote transmit and receive, respectively.}
 	\label{FD_basestation} \vspace{-2mm}
\end{figure}
\section{System Model} \label{sytem_model}
We consider a single-cell mmWave FD system consisting of one hybrid FD BS serving $J$ DL and $K$ UL fully digital multi-antenna users, as shown in Fig. \ref{FD_basestation}. We assume perfect channel state information (CSI)\footnote{The CSI of the mmWave FD systems can be acquired similarly as in \cite{chen2017structured} for the mmWave HD system and it is part of the ongoing research \cite{roberts2021millimeter}.}. The FD BS is assumed to have $M_t$ transmit and $N_r$ receive RF chains, and $M_0$ transmit and $N_0$ receive antennas. Let $\mathcal{U} = \{1,...,K\}$ and $\mathcal{D} = \{1,...,J\} $ denote the sets containing the indices of $K$ UL and $J$ DL users, respectively. Let $M_{k}$ and $N_{j}$ denote the number of transmit and receive antennas for $k$-th UL and $j$-th DL user, respectively. We consider a multi-stream approach and the number of data streams for $k$-th UL and $j$-th DL user are denoted as $u_k$ and $v_j$, respectively. Let $\bm{U}_k \in \mathbb{C}^{M_k \times u_k}$ and $\bm{V}_j \in \mathbb{C}^{M_t \times v_j}$ denote the precoders for white unitary variance data streams $\bm{s}_{k}\in \mathbb{C}^{u_k \times 1}$ and $\bm{s}_{j} \in \mathbb{C}^{v_j \times 1}$, respectively. Let $\bm{G}_{RF} \in \mathbb{C}^{M_0 \times M_t}$ and $\bm{F}_{RF}  \in \mathbb{C}^{N_0 \times N_r}$ denote the fully connected analog beamformer and combiner at the FD BS, respectively. Let $\mathcal{P} = \{1,e^{i 2 \pi/n_{ps}},...,e^{i 2 \pi n_{ps}-1/n_{ps}}\}$ denote the set of $n_{p_s}$ possible discrete values that the phasors at the analog stage can assume on unit-circle. 

For HYBF with the unit-modulus constraint, we define the quantizer function $\mathbb{Q}_P(\cdot)$ to quantize the unit-modulus phasors of analog beamformer $\bm{G}_{RF}$ and combiner $\bm{F}_{RF}$ such that $\mathbb{Q}_P(\angle \bm{G}_{RF}(m,n)) \in \mathcal{P}$ and $\mathbb{Q}_P(\angle \bm{F}_{RF}(m,n)) \in \mathcal{P}$, $\forall m,n$. For HYBF with amplitude control, the phase part is still quantized with $\mathbb{Q}_P(\cdot)$ and belongs to $\mathcal{P}$. Let $\mathcal{A}=\{a_0,....,a_{A-1}\}$ denote the set of $A$ possible values that the amplitudes can assume. Let $\mathbb{Q}_{A}(\cdot)$ denote the quantizer function to quantize the amplitudes of $\bm{G}_{RF}$ and $\bm{F}_{RF}$ such that $\mathbb{Q}_{A}(|\bm{G}_{RF}(m,n)|) \in \mathcal{A}$ and $\mathbb{Q}_{A}(|\bm{F}_{RF}(m,n)|) \in \mathcal{A}, \; \forall m,n$. A complex number $\bm{G}_{RF}(m,n)$ with amplitude in $\mathcal{A}$ and phase part in $\mathcal{P}$ can be written as $\bm{G}_{RF}(m,n) = \mathbb{Q}_{A}(|\bm{G}_{RF}(m,n)|) \mathbb{Q}_P(\angle \bm{G}_{RF}(m,n))$.
The thermal noise vectors for FD BS and $j$-th DL user are denoted as $ \bm{n}_0\sim  \mathcal{CN}(0,\sigma_0^2 \bm{I}_{N_0})$ and $\bm{n}_j \sim \mathcal{CN}(0,\sigma_j^2 \bm{I}_{N_j})$, respectively. Let $\bm{c_k}$ and $\bm{e}_j$ denote the LDR noise vectors for $k$-th UL and $j$-th DL user, respectively, which can be modelled as \cite{day2012full}
\begin{equation}  
    \bm{c}_k \sim \mathcal{CN}\Big(\bm{0}_{M_k \times 1}, k_k\; \mbox{diag}(\bm{U}_k  \bm{U}_k^H )\Big),\label{tx_ldr_UL}
\end{equation}

\begin{equation}
    \bm{e_j} \sim \mathcal{CN}\Big({\bm{0}_{N_j \times 1} , \beta_j\; \mbox{diag}(\bm{\Phi}_j)}\Big), \label{rx_ldr_DL}
\end{equation}
where $k_k \ll 1$, $\beta_j \ll 1, \bm{\Phi}_j =  \mbox{Cov}(\bm{r}_j)$ and $\bm{r}_j$ denotes the undistorted signal received by $j$-th DL user.
Let $\bm{c_0}$ and $\bm{e_0}$ denote the LDR noise vectors in transmission and reception for FD BS, respectively. We model them as

\begin{equation} 
   \bm{c}_0 \sim \mathcal{CN}\Big(\bm{0}_{M_0 \times 1}, k_0\; \mbox{diag}(\sum_{n \in \mathcal{D}} \bm{G}_{RF} \bm{V}_n  \bm{V}_n^H \bm{G}_{RF}^H)\Big) ,
   \label{tx_ldr}
\end{equation}
\begin{equation}
    \bm{e_0} \sim \mathcal{CN}\Big({\bm{0}_{N_r \times 1} , \beta_0\; \mbox{diag}(\bm{\Phi}_0)}\Big),   \label{Rx_ldr}
\end{equation}
where $ k_0 \ll 1 $, $\beta_0 \ll 1, \bm{\Phi}_0 = \mbox{Cov}(\bm{r}_0)$ and $\bm{r}_0$ denotes the undistorted signal received by FD BS after the analog combiner $\bm{F}_{RF}$. Note that \eqref{tx_ldr} extends the transmit LDR noise model from \cite{day2012full} to the case of a hybrid transmitter. For the hybrid receiver at the mmWave FD BS, ADCs, the most dominant sources of receive LDR noise, are placed after the analog combiner $\bm{F}_{RF}$. Consequently, $\bm{e_0}$ in \eqref{Rx_ldr} considers the undistorted signal received after the analog combiner. We remark that the extension presented in \eqref{tx_ldr}-\eqref{Rx_ldr} is slightly simplified. In practice, as some circuitry might be shared among multiple antennas, it can lead to some correlation.
 
Let $\bm{y}$ and $\bm{y}_j$ denote the signals received by the FD BS and $j$-th DL user, respectively, which can be written as
 
\begin{equation}
\begin{aligned}
       \bm{y} = &\bm{F}_{RF}^H \sum_{k \in \mathcal{U}} \bm{H}_k   \bm{U}_k  \bm{s}_{k} +  \bm{F}_{RF}^H \sum_{k \in \mathcal{U}} \bm{H}_k \bm{U}_k \bm{c}_k  + \bm{F}_{RF}^H \bm{n}_0\\& +  \bm{F}_{RF}^H \bm{H}_0  \sum_{j \in \mathcal{D}} \bm{G}_{RF} \bm{V}_j \bm{s}_{j} +  \bm{F}_{RF}^H \bm{H}_0 \bm{c}_0  +  \bm{e}_0,
\end{aligned} \label{BS_side}
\end{equation}
\begin{equation}
\begin{aligned}
         \bm{y}_j = \bm{H}_j  \sum_{n \in \mathcal{D}} &\bm{G}_{RF} \bm{V}_n  \bm{s}_n  + \bm{H}_j  \sum_{n \in \mathcal{D}} \bm{G}_{RF} \bm{V}_n \bm{c}_0  + \bm{e}_j + \bm{n}_j \\& + \sum_{k \in \mathcal{U}} \bm{H}_{j,k}  \bm{U}_k  \bm{s}_{k} + \sum_{k \in \mathcal{U}} \bm{H}_{j,k} \bm{c}_k .
\end{aligned} \label{Rx_side}
\end{equation}
The matrices $\bm{H}_k \in \mathbb{C}^{N_0 \times M_k}$ and $\bm{H}_j \in \mathbb{C}^{N_j \times M_0}$ denote channel response from the $k$-th UL user to BS and from the BS to $j$-th DL user, respectively. The matrices $\bm{H}_0\in  \mathbb{C}^{N_0 \times M_0}$ and $\bm{H}_{j,k} \in  \mathbb{C}^{N_j \times M_k}$ denote SI channel response for FD BS and cross-interference channel response between $k$-th UL and $j$-th DL users, respectively. At the mmWave, the channel response $\bm{H}_k$ can be modelled as\cite{satyanarayana2018hybrid}
\begin{equation} 
    \bm{H}_k = \sqrt{\frac{M_k N_0}{N_c N_{p}}} \sum_{n_c = 1}^{N_c} \sum_{n_p = 1}^{N_p} \alpha_{k}^{(n_p,n_c)} \bm{a}_r(\phi_k^{n_p,n_c}) \bm{a}_{t}^T(\theta_k^{n_p,n_c}), \label{channel_model}
\end{equation} 
where $N_c$ and $N_{p}$ denote the number of clusters and number of rays (Figure 1 \cite{satyanarayana2018hybrid}), respectively, and $\alpha_{k}^{(n_p,n_c)} \sim \mathcal{CN}(0,1)$ denotes a complex Gaussian random variable whose amplitude and phase are Rayleigh and uniformly
distributed, respectively. The vectors $\bm{a}_r(\phi_k^{n_p,n_c})$ and  $\bm{a}_t^T(\theta_k^{n_p,n_c})$ denote the receive and transmit antenna array response with angle of arrival (AoA) $\phi_k^{n_p,n_c}$ and angle of departure (AoD) $\theta_k^{n_p,n_c}$, respectively. The channel matrices $\bm{H}_j$ and $\bm{H}_{j,k}$ can be modelled similarly as in \eqref{channel_model}. The SI channel can be modelled as \cite{satyanarayana2018hybrid}
\begin{equation} \label{SI_Channel}
    \bm{H}_0 = \sqrt{\frac{\kappa}{\kappa+1}} \bm{H}_{LoS} + \sqrt{\frac{1}{\kappa+1}} \bm{H}_{ref},
\end{equation}
where $\kappa$ denotes the Rician factor, and the matrices $\bm{H}_{LoS}$ and $\bm{H}_{ref}$ denote the line-of-sight (LoS) and reflected contributions, respectively. The channel matrix $\bm{H}_{ref}$ can be modelled as \eqref{channel_model} and element of $\bm{H}_{LoS}$ at the $m$-th row and $n$-th column can be modelled as \cite{satyanarayana2018hybrid}
\begin{equation} \label{SI_LOS_model}
    \bm{H}_{LoS}(m,n) = \frac{\rho}{r_{m,n}} e^{-i 2 \pi \frac{r_{m,n}}{\lambda}}.
\end{equation}
where $\rho$ denotes the power normalization constant to assure $\mathbb{E}(||\bm{H}_{LoS}(m,n)||_F^2)=M_0 N_0$ and $\lambda$ denotes the wavelength. The scalar $r_{m,n}$ denotes distance between the $m$-th receive and $n$-th transmit antenna, which depends on the transmit and receive array geometry (9) \cite{satyanarayana2018hybrid}. 
The aforementioned notations are summarized in Table \ref{table_parametri_definiti}.
 
\begin{table}      
\centering   \caption{Notations}
    \resizebox{8cm}{!}{%
    \begin{tabular}{|p{15mm}|p{65mm}|}
        \hline  
        \mbox{$M_{t}$} & Number of transmit RF chains for the BS \\
        \hline
         \mbox{$N_{r}$}  & Number of receive RF chains for the FD BS \\
         \hline
          \mbox{$M_{0}$} & Number of transmit antennas for the BS  \\
          \hline
         \mbox{$N_{0}$} & Number of receive antennas for the BS  \\
         \hline
        \mbox{$M_{k}$} & Number of transmit antennas for UL user $k$\\ 
         \hline
        \mbox{$N_j$} & Number of receive antennas for  DL user $j$ \\ 
         \hline
         \mbox{$\bm{U}_k$} & Digital beamformer for  UL user $k$  \\ 
         \hline
         \mbox{$\bm{V}_j$} & Digital beamformer for  DL user $j$ \\ 
          \hline
         \mbox{$\bm{G}_{RF}$} & Analog beamformer for the FD BS \\
          \hline
           \mbox{$\bm{F}_{RF}$} & Analog combiner for the FD BS \\
          \hline
         \mbox{$\bm{c}_k$} & Transmit LDR noise from UL user $k$  \\ 
          \hline
         \mbox{$\bm{c}_0$} &  Transmit LDR noise from the FD BS  \\ 
          \hline
         \mbox{$\bm{e}_0$} &  Receive LDR noise at the FD BS   \\  
          \hline
         \mbox{$\bm{e}_j$}  &  Receive LDR noise at the DL user $j$ \\ 
          \hline
         \mbox{$\bm{n}_0$} &  Thermal noise at the FD BS \\  
          \hline
         \mbox{$\bm{n}_j$}  & Thermal noise at the DL user $j$ \\  
          \hline
         \mbox{$\bm{H}_0$} & SI channel\\  
          \hline
         \mbox{$\bm{H}_{k}$}  & Channel between the BS and UL user $k$\\
          \hline
         \mbox{$\bm{H}_j$} & Channel between the BS and DL user $j$\\ 
          \hline
         \mbox{$\bm{H}_{j,k}$}  & Cross-interference channel between UL user $k$ and DL user $j$  \\
         \hline
         	$\preceq$  & Element-wise inequality  \\
         \hline
    \end{tabular}}    \label{table_parametri_definiti}
\end{table}
 
\subsection{Problem Formulation}
Let $\overline{k}$ and $\overline{j}$ denote the indices in sets $\mathcal{U}$ and $\mathcal{D}$ without the elements $k$ and $j$, respectively. The received (signal plus) interference and noise covariance matrices from UL user $k \in \mathcal{U}$ at the BS and by the DL user $j \in \mathcal{D}$ are denoted as ($\bm{R}_k$)  $\bm{R}_{\overline{k}}$ and ($\bm{R}_j$) $\bm{R}_{\overline{j}}$, respectively. Let $\bm{T}_k$, $\forall k \in \mathcal{U}$, and $\bm{Q}_j$, $\forall j \in \mathcal{D}$, defined as 
 \begin{subequations}\label{reparametarization}
 \begin{equation}
\bm{T}_k = \bm{U}_k \bm{U}_k^H, 
\end{equation}
\begin{equation}
 \bm{Q}_j = \bm{G}_{RF} \bm{V}_j \bm{V}_j^H \bm{G}_{RF}^H, 
\end{equation}
 \end{subequations}
denote the transmit covariance matrices from UL user $k \in \mathcal{U}$ and DL user $j \in \mathcal{D}$, respectively. By considering the distortions from non-ideal hardware with the extended LDR noise model, cross-interference, interference and SI, the received covariance matrices at the BS after the analog combiner, i.e., $\bm{R}_k$ and $\bm{R}_{\overline{k}}$, and at the DL user $j \in \mathcal{D}$, i.e., $\bm{R}_j$ and  $\bm{R}_{\overline{j}}$, can be written as \eqref{covariance_matrices}, shown at the top of the next page. In \eqref{covariance_matrices}, $\bm{S}_{k}$ and $\bm{S}_{j}$ denote the useful received signal covariance matrices from $k$-th UL user at the FD BS and by $j$-th DL user, respectively. The undistorted received covariance matrices can be recovered from \eqref{covariance_matrices} as $\bm{\Phi}_0 = \bm{R}_k$, with $\beta_0=0$, and $\bm{\Phi}_j = \bm{R}_j$, with $\beta_j=0$.

\begin{figure*}[!t]  
\begin{subequations} \label{covariance_matrices}
\begin{equation}  
\begin{aligned}
\bm{R}_k =   \underbrace{\bm{F}_{RF}^H \bm{H}_k  \bm{T}_k \bm{H}_k^H  \bm{F}_{RF}}_{\triangleq \bm{S}_{k}} + \sum_{\substack{i \in \mathcal{U} \\i\neq k}} \bm{F}_{RF}^H \bm{H}_i \bm{T}_i \bm{H}_i^H \bm{F}_{RF} + \sum_{i \in \mathcal{U}} k_i \bm{F}_{RF}^H  & \bm{H}_i \mbox{diag}\Big(\bm{T}_i\Big) \bm{H}_i^H \bm{F}_{RF}  + \sigma_0^2\bm{I}_{N_0} + \beta_0 \mbox{diag}\Big(\bm{\Phi}_0\Big) \\& + \bm{F}_{RF}^H \bm{H}_0\Big(\sum_{n \in \mathcal{D}}^{}\bm{Q}_n + k_0 \mbox{diag}\Big(\sum_{n \in \mathcal{D}} \bm{Q}_n\Big)\Big) \bm{H}_0^H  \bm{F}_{RF},
\end{aligned} 
\end{equation}
\begin{equation}
\begin{aligned}
\bm{R}_{j} = & \underbrace{\bm{H}_j \bm{Q}_j \bm{H}_j^H}_{\triangleq \bm{S}_j}  + \bm{H}_j \sum_{\substack{n \in \mathcal{D}\\ n \neq j}} \bm{Q}_n \bm{H}_j^H + k_0 \bm{H}_j \mbox{diag}\Big(\sum_{n \in \mathcal{D}} \bm{Q}_n\Big) \bm{H}_j^H +\sigma_j^2 \bm{I}_{N_j} + \sum_{i \in \mathcal{U}} \bm{H}_{j,i} \Big( \bm{T}_i  + k_i \mbox{diag}\Big(\bm{T}_i\Big) \Big) \bm{H}_{j,i}^H  + \beta_j \mbox{diag}\Big(\bm{\Phi}_j \Big)  ,
\end{aligned}
\end{equation}
\begin{equation}
\begin{aligned}
\bm{R}_{\overline{k}} = \bm{R}_k- \bm{S}_k, \quad
\bm{R}_{\overline{j}} = \bm{R}_j - \bm{S}_j.
\end{aligned}
\end{equation}
\end{subequations} \hrulefill
\end{figure*} 

The WSR maximization problem with respect to the digital beamformers, analog beamformer and combiner with amplitudes in $\mathcal{A}$ and phase part in $\mathcal{P}$, under the joint sum-power and per-antenna power constraints, can be stated as

\begin{subequations}\label{problem_statement}
\begin{equation}
    \underset{\substack{\bm{U},\bm{V},\\ \bm{G}_{RF}\bm{F}_{RF}}}{\text{max}} \quad \hspace{-2mm}\sum_{k \in \mathcal{U}} w_k \mbox{lndet}\Big( \bm{R}_{\overline{k}}^{-1} \bm{R}_k \Big) + \sum_{j \in \mathcal{D}} w_j \mbox{lndet}\Big( \bm{R}_{\overline{j}}^{-1} \bm{R}_j\Big)
\end{equation} \label{WSR}
\begin{equation}
\text{s.t.} \quad \mbox{diag}\Big( \bm{U}_k \bm{U}_k^H \Big) 	\preceq \bm{\Lambda_k}, \quad \forall k \in \mathcal{U}, \label{c1}
\end{equation}
\begin{equation}
 \quad \quad  \quad \quad \mbox{diag}{\Big(  \sum_{j \in \mathcal{D}}  \bm{G}_{RF} \bm{V}_j  \bm{V}_j^H \bm{G}_{RF}^H\Big) 	\preceq \bm{\Lambda}_0},\label{c2}
\end{equation}
\begin{equation}
 \quad \quad  \mbox{Tr}\Big(\bm{U}_k \bm{U}_k^H \Big) \leq  \alpha_k,  \quad  \forall k \in \mathcal{U}, \label{c3}
\end{equation}
\begin{equation}
\quad \quad \mbox{Tr} \Big(\sum_{j \in \mathcal{D}} \bm{G}_{RF} \bm{V}_j  \bm{V}_j^H \bm{G}_{RF}^H \Big) \leq  \alpha_0. \label{c4}
\end{equation}
\begin{equation}
 \quad \quad  \quad  \angle \bm{G}_{RF}(m,n) \in \mathcal{P}, \;\mbox{and}\; |\bm{G}_{RF}(m,n)| \in \mathcal{A}, \; \forall \; m,n, \label{c6}
\end{equation} 
\begin{equation}
\quad \quad  \quad  \angle \bm{F}_{RF}(i,j) \in \mathcal{P}, \;\mbox{and}\;|\bm{F}_{RF}(i,j)| \in \mathcal{A}, \quad  \forall \; i,j. \label{c7}
\end{equation}
\end{subequations}
The scalars $w_k$ and $w_j$ denote rate weights for the UL user $k$ and DL user $j$, respectively. The diagonal matrices $\bm{\Lambda_k}$ and $\bm{\Lambda_0}$ denote per-antenna power constraints for the $k$-th UL user and FD BS, respectively, and the scalars $\alpha_k$ and $\alpha_0$ denote their sum-power constraint.
The collections of digital UL and DL beamformers are denoted as $\bm{U}$ and $ \bm{V}$, respectively. For unit-modulus HYBF, the constraints in $\eqref{c6}-\eqref{c7}$ on the amplitude part become unit-modulus.

\emph{Remark 1:} Note that the rate achieved with \eqref{WSR} is not affected by the digital receivers if they are chosen as the MMSE combiners, see e.g., $(4)-(9)$ \cite{christensen2008weighted} for more details. For WSR maximization, only the analog combiner has to considered in the optimization problem as it affects the size of the received covariance matrices from UL users, i.e., the UL rate.

 \vspace{-0.6mm}

\section{Minorization-Maximization} \label{simplificazione_problem}
Problem \eqref{WSR} is non-concave in the transmit covariance matrices $\bm{T}_k $ and $\bm{Q}_j$ due to the interference terms and searching its globally optimum solution is very challenging. In this section, we present the minorization-maximization optimization method \cite{StoicaSelen:SPmag0104} for solving \eqref{WSR} to a local optimum.

The WSR maximization problem \eqref{WSR} will be reformulated at each iteration as a concave reformulation with its minorizer, using the difference-of-convex (DC) programming \cite{kim2011optimal} in terms of the variable to be updated, while the other variables will be kept fixed. To proceed, note that the WSR in \eqref{WSR} can be written with the weighted-rate (WR) of user $k \in \mathcal{U}$, user $j \in \mathcal{D}$, WSRs for $\overline{k}$ and $\overline{j}$ as

\begin{equation}
    \mbox{WSR} = \underbrace{\mbox{WR}_{k}^{UL} + \mbox{WSR}_{\overline{k}}^{UL}}_{\triangleq \mbox{WSR}^{UL}} + \underbrace{\mbox{WR}_{j}^{DL}+\mbox{WSR}_{\overline{j}}^{DL}}_{\triangleq \mbox{WSR}^{DL}},
\label{eqWSR}
\end{equation}
where $\mbox{WSR}^{UL}$ and $\mbox{WSR}^{DL}$ denote the WSR in UL and DL, respectively. 
Considering the dependence of the transmit covariance matrices, only $\mbox{WR}_{k}^{UL}$ is concave in $\bm{T}_k$, meanwhile $\mbox{WSR}_{\overline{k}}^{UL}$ and $\mbox{WSR}^{DL}$ are non-concave in $\bm{T}_k$, when $\bm{T}_{\overline{k}}$ and $\bm{Q}_{j}$, $\forall j \in \mathcal{D}$, are fixed. 
Similarly, only $\mbox{WSR}_{j}^{DL}$ is concave in $\bm{Q}_j$ and non-concave in $\mbox{WSR}_{\overline{j}}^{DL}$ and $\mbox{WSR}^{UL}$, when $\bm{Q}_{\overline{j}}$ and $\bm{T}_{k}$, $\forall k \in \mathcal{U}$, are fixed. Since a linear function is simultaneously convex and concave, DC programming introduces the first order
Taylor series expansion of $\mbox{WSR}_{\overline{k}}^{UL}$ and $\mbox{WSR}^{DL}$ in $\bm{T}_k$, around $\hat{\bm{T}}_k$ (i.e. around all $\bm{T}_k$), and of $\mbox{WSR}_{\overline{j}}^{DL}$ and $\mbox{WSR}^{UL}$ in $\bm{Q}_j$, around $\hat{\bm{Q}}_j$ (i.e. around all $\bm{Q}_j$). Let $\hat{\bm{T}}$ and $\hat{\bm{Q}}$ denote the set containing all such $\hat{\bm{T}}_k$ and $\hat{\bm{Q}}_j$, respectively. Let $\hat{\bm{R}_{k}}(\bm{\hat{T},\hat{Q}})$, $\hat{\bm{R}_{\overline{k}}}(\bm{\hat{T},\hat{Q}})$, $\hat{\bm{R}_{j}}(\bm{\hat{T},\hat{Q}})$, and $\hat{\bm{R}_{\overline{j}}}(\bm{\hat{T},\hat{Q}})$ denote the covariance matrices $\bm{R}_{k},\bm{R}_{\overline{k}}, \bm{R}_{j}$ and $\bm{R}_{\overline{j}}$ as a function of $\hat{\bm{T}}$ and $\hat{\bm{Q}}$, respectively. The linearized tangent expressions for each communication link by computing the gradients 
\begin{subequations}
\begin{equation}
    \hat{\bm{A}}_{k} = - \frac{\partial \mbox{WSR}_{\overline{k}}^{UL}}{\partial \bm{T}_k}\Bigr|_{\hat{\bm{T}},\hat{\bm{Q}}}
    ,\quad   \hat{\bm{B}}_{k}  = - \frac{\partial \mbox{WSR}^{DL}}{\partial \bm{T}_k}\Bigr|_{\hat{\bm{T}},\hat{\bm{Q}}}, 
    \label{grad_UL}
\end{equation}
\begin{equation}
    \hat{\bm{C}}_{j}  = - \frac{\partial \mbox{WSR}_{\overline{j}}^{DL}}{\partial \bm{Q}_j} \Bigr|_{\hat{\bm{T}},\hat{\bm{Q}}}
    , \quad
    \hat{\bm{D}}_{j} = - \frac{\partial \mbox{WSR}^{UL}}{\partial \bm{Q}_j}\Bigr|_{\hat{\bm{T}},\hat{\bm{Q}}},
    \label{grad_DL}
\end{equation}
\end{subequations}
with respect to the transmit covariance matrices $\bm{T}_k$ and $\bm{Q}_j$ can be written as
\begin{subequations}
\begin{equation}
   \underline{\mbox{WSR}}_{\overline{k}}^{UL}\Big(\bm{T}_k,\hat{\bm{T}},\hat{\bm{Q}}\Big) = \mbox{WSR}_{\overline{k}}^{UL}(\hat{\bm{T}},\hat{\bm{Q}}) 
     - \mbox{Tr}\Big(\Big( \bm{T_k} - \hat{\bm{T}}_k\Big) \hat{\bm{A}}_k\Big),
     \label{eqAk}
\end{equation}
\begin{equation}
 \underline{\mbox{WSR}}^{DL}\Big(\bm{T}_k,\hat{\bm{T}},\hat{\bm{Q}}\Big) = \mbox{WSR}^{DL}(\hat{\bm{T}},\hat{\bm{Q}})
     -\mbox{Tr}\Big(\Big(\bm{T_k} - \hat{\bm{T}}_k\Big) \hat{\bm{B}}_k\Big),
          \label{eqBk}
\end{equation}
\begin{equation}
       \underline{\mbox{WSR}}_{\overline{j}}^{DL}\Big(\bm{Q}_j,\hat{\bm{Q}},\hat{\bm{T}}\Big) = \mbox{WSR}_{\overline{j}}^{DL}(\hat{\bm{T}},\hat{\bm{Q}}) 
       - \mbox{Tr}\Big(\Big( \bm{Q_j} - \hat{\bm{Q}}_j\Big) \hat{\bm{C}}_j\Big),
     \label{eqCk}
\end{equation}
\begin{equation}
 \underline{\mbox{WSR}}^{UL}\Big(\bm{Q}_j,\hat{\bm{Q}},\hat{\bm{T}}\Big)=\mbox{WSR}^{UL}(\hat{\bm{T}},\hat{\bm{Q}}) 
    - \mbox{Tr}\Big(\Big( \bm{Q}_j - \hat{\bm{Q}}_j \Big) \hat{\bm{D}}_j\Big).
         \label{eqDk}
\end{equation}
\end{subequations}

\begin{figure*}[!t] 
\begin{subequations}  
\begin{equation}
\begin{aligned}
\hat{\bm{A}}_{k} = \sum_{\substack{i \in \mathcal{U},i \neq k}}  w_i \Big( \bm{H}_k^H \bm{F}_{RF}  \Big[  \hat{\bm{R}_{\overline{i}}}(\bm{\hat{T},\hat{Q}})^{-1}\hspace{-2mm} - \hat{\bm{R}_i}(\bm{\hat{T},\hat{Q}})^{-1}  &  - \beta_0 \; \mbox{diag}\Big(\hat{\bm{R}_{\overline{i}}}(\bm{\hat{T},\hat{Q}})^{-1}  - \hat{\bm{R}_i}(\bm{\hat{T},\hat{Q}})^{-1} \Big) \Big] \bm{F}_{RF}^H \bm{H}_k \\& - k_i\; \mbox{diag}\Big(\bm{H}_k^H \bm{F}_{RF} \Big( \hat{\bm{R}_{\overline{i}}}(\bm{\hat{T},\hat{Q}})^{-1} - \hat{\bm{R}_{i}}(\bm{\hat{T},\hat{Q}})^{-1} \Big) \bm{F}_{RF}^H \bm{H}_k \Big) \Big),
\end{aligned}
\end{equation} 
\begin{equation}
\begin{aligned}
\hat{\bm{B}}_{k} =  \sum_{\substack{l \in \mathcal{D}}}w_l \Big( \bm{H}_{l,k}^H  \Big[\hat{\bm{R}_{\overline{l}}}(\bm{\hat{T},\hat{Q}})^{-1}   -  \hat{\bm{R}_l}(\bm{\hat{T},\hat{Q}})^{-1}  & - \beta_j\; \mbox{diag}\Big(\hat{\bm{R}_{\overline{l}}}(\bm{\hat{T},\hat{Q}})^{-1} - \hat{\bm{R}_l}(\bm{\hat{T},\hat{Q}})^{-1}\Big) \Big]  \bm{H}_{l,k} \\& - k_k\; \mbox{diag}\Big(\bm{H}_{l,k}^H \Big( \hat{\bm{R}_{\overline{l}}}(\bm{\hat{T},\hat{Q}})^{-1} - \bm{R}_{l}(\bm{\hat{T},\hat{Q}})^{-1}\Big)  \bm{H}_{l,k}\Big)\Big),
\end{aligned}
\end{equation} 
\begin{equation} 
\begin{aligned}
\hat{\bm{C}}_{j} = \sum_{\substack{n \in \mathcal{D}, n \neq j}} w_n \Big(\bm{H}_n^H  \Big[\hat{\bm{R}}_{\overline{n}}(\bm{\hat{T},\hat{Q}})^{-1} - \hat{\bm{R}}_n(\bm{\hat{T},\hat{Q}})^{-1} & -   \beta_n\; \mbox{diag}\Big(\hat{\bm{R}_{\overline{n}}}(\bm{\hat{T},\hat{Q}})^{-1}  - \hat{\bm{R}_n}(\bm{\hat{T},\hat{Q}})^{-1}\Big)\Big]\bm{H}_n  \\& - k_0\; \mbox{diag}\Big(\bm{H}_n^H  (\hat{\bm{R}}_{\overline{n}}(\bm{\hat{T},\hat{Q}})^{-1} - \hat{\bm{R}}_{n}(\bm{\hat{T},\hat{Q}})^{-1}\Big)  \bm{H}_n \Big),
\end{aligned}
\end{equation} 
\begin{equation} 
\begin{aligned}
\hat{\bm{D}}_{j} =  \sum_{\substack{m \in \mathcal{U}}}  w_m \Big( \bm{H}_0^H \bm{F}_{RF} \Big[ \hat{\bm{R}}_{\overline{m}}(\bm{\hat{T},\hat{Q}})^{-1}  - \hat{\bm{R}}_m(\bm{\hat{T},\hat{Q}})^{-1} & - \beta_0\; \mbox{diag}\Big(\hat{\bm{R}}_{\overline{m}}(\bm{\hat{T},\hat{Q}})^{-1}  - \hat{\bm{R}}_m(\bm{\hat{T},\hat{Q}})^{-1}\Big) \Big] \bm{F}_{RF}^H \bm{H}_0   \\&- k_0\; \mbox{diag}\Big(\bm{H}_0^H \bm{F}_{RF}   \Big(\hat{\bm{R}_{\overline{m}}}(\bm{\hat{T},\hat{Q}})^{-1} - \hat{\bm{R}}_m(\bm{\hat{T},\hat{Q}})^{-1}\Big) \bm{F}_{RF}^H  \bm{H}_0\Big) \Big),
\end{aligned}
\end{equation} \label{gradients}
\end{subequations} \hrulefill
\end{figure*}
We remark that the tangent expressions \eqref{eqAk}-\eqref{eqDk}  constitute a touching lower bound for $\mbox{WSR}_{\overline{k}}^{UL},\mbox{WSR}_{\overline{j}}^{DL}, \mbox{WSR}^{DL}$ and $\mbox{WSR}^{UL}$, respectively. Hence, the DC programming approach is also a minorization-maximization approach, regardless of the restatement of the transmit covariance matrices $\bm{T}_k$ and $\bm{Q}_j$ as a function of the beamformers.  
\begin{thm} \label{thm_grad}
The gradients $\hat{\bm{A}}_{k}$ and $\hat{\bm{B}}_{k}$ which linearize $\mbox{WSR}_{\overline{k}}^{UL}$ and $\mbox{WSR}^{DL}$, respectively, with respect to $\bm{T}_k$,  $\forall k \in \mathcal{U}$, and the gradients $\hat{\bm{C}}_{j}$ and $\hat{\bm{D}}_{j}$ which linearize $\mbox{WSR}_{\overline{j}}^{DL}$ and $\mbox{WSR}^{UL}$, respectively, with respect to $\bm{Q}_j$, $\forall j \in \mathcal{D}$, with the first order Taylor series expansion are given in \eqref{gradients}. 
\end{thm}
\begin{proof}
Please see Appendix \ref{grad_proof}.
\end{proof}

\begin{figure*}[!t] 
\begin{equation} 
\begin{aligned}
& \underset{\substack{\bm{U},\bm{V} \\\bm{G}_{RF}\bm{F}}_{RF}}{\text{max}} \quad \hspace{-2mm}\sum_{k \in \mathcal{U}} w_k \mbox{lndet}\Big(\bm{I} + \bm{U}_k^H  \bm{H}_k^H \bm{F}_{RF} \bm{R}_{\overline{k}}^{-1} \bm{F}_{RF}^H \bm{H}_k \bm{U}_k  \Big) - \mbox{Tr} \Big( \bm{U}_k^H \Big(\hat{\bm{A}}_{k} + \hat{\bm{B}}_{k}\Big) \bm{U}_k  \Big)+\\& \quad \quad \quad \sum_{j \in \mathcal{D}} w_j  
 \mbox{lndet}\Big( \bm{I} +  \bm{V}_j^H \bm{G}_{RF}^H \bm{H}_j^H  \bm{R}_{\overline{j}}^{-1} \bm{H}_j \bm{G}_{RF} \bm{V}_j \Big) -  \mbox{Tr}\Big( \bm{V}_j^H   \bm{G}_{RF}^H \Big(\hat{\bm{C}}_{j} + \hat{\bm{D}}_{j}\Big) \bm{G}_{RF} \bm{V}_j \Big) \\
 & \quad \quad \quad \quad  \quad \quad \quad \quad   \quad \quad  \quad \quad  \quad \quad \text{s.t.} \quad \quad \eqref{c1}-\eqref{c7} 
\end{aligned} \label{WSR_convex}
\end{equation}  \hrulefill
\end{figure*}
\vspace{-3mm}
\subsection{Concave Reformulation}
In this section, we simplify the non-concave WSR maximization problem \eqref{WSR}. By using the gradients \eqref{gradients}, \eqref{WSR} can be reformulated as \eqref{WSR_convex}, given at the top of the next page.

\begin{lemma}
The WSR maximization problem \eqref{WSR} for a single-cell mmWave FD system with multi-antenna users reformulated at each iteration with its first-order Taylor series expansion as in \eqref{WSR_convex} is a concave reformulation for each link.
\end{lemma}
\begin{proof}
The optimization problem \eqref{WSR} restated as in \eqref{WSR_convex} for each link is made of a concave part, i.e., log($\cdot$), and a linear part, i.e., $\mbox{Tr}(\cdot)$. Since a linear function is simultaneously concave and non-concave, \eqref{WSR_convex} results to be concave for each link.
\end{proof}

\emph{Remark 2:} The problem \eqref{WSR} and its reformulated version \eqref{WSR_convex} have the same Karush–Kuhn–Tucker (KKT) conditions and therefore any sub-optimal (optimal) solution of \eqref{WSR_convex} is also sub-optimal (optimal) for \eqref{WSR}. 

Let $\bm{\Psi}_0 = \mbox{diag}([\psi_1,...,\psi_{M_0}])$ and $\bm{\Psi}_k = \mbox{diag}([\psi_{k,1},...,\psi_{k,M_k}])$, denote diagonal matrices containing the Lagrange multipliers associated with per-antenna power constraints for the FD BS and UL user $k$, respectively. Let $l_0$ and $l_1 ,...,l_K$ denote the Lagrange multipliers associated with the sum-power constraint for FD BS and $K$ UL users, respectively. 
Let $\bm{\Psi}$ denote the collection of Lagrange multipliers associated with the per-antenna power constraints, i.e., $\bm{\Psi}_0$ and $\bm{\Psi}_k, \forall k \in \mathcal{U}$. Let $\bm{L}$ denote the collection of Lagrange multipliers associated with the sum-power constraints.
Augmenting the linearized WSR maximization problem \eqref{WSR_convex} with the sum-power and practical per-antenna power constraints, yields the Lagrangian \eqref{AUG_Largrangian}, given at the top of this page. In \eqref{AUG_Largrangian}, unconstrained analog beamformer and combiner are assumed and their constraints will be  incorporated later.

\begin{figure*}[!t]
\begin{equation} 
\begin{aligned}
 \mathcal{L}(\bm{U}, &\bm{V},\bm{G}_{RF},\bm{F}_{RF},\bm{\Psi},\bm{L})= \sum_{l=0}^{K} l_l \alpha_l + \mbox{Tr}\Big(\bm{\Psi_0 \bm{\Lambda_0}}\Big) + \sum_{u \in \mathcal{U}} \mbox{Tr}\Big(\bm{\Psi_u \bm{\Lambda_u}}\Big) \\
 &+ \sum_{k \in \mathcal{U}} w_k\mbox{lndet}\Big(\bm{I} + \bm{U}_k^H  \bm{H}_k^H \bm{F}_{RF}  \bm{R}_{\overline{k}}^{-1}  \bm{F}_{RF}^H \bm{H}_k \bm{U}_k  \Big) -  \mbox{Tr}\Big( \bm{U}_k^H \Big(\hat{\bm{A}}_{k} + \hat{\bm{B}}_{k} + l_k \bm{I} + \bm{\Psi}_k \Big) \bm{U}_k  \Big) \\
 & + \sum_{j\in \mathcal{D}} w_j \mbox{lndet}\Big(\bm{I} +  \bm{V}_j^H \bm{G}_{RF}^H \bm{H}_j^H  \bm{R}_{\overline{j}}^{-1} \bm{H}_j \bm{G}_{RF} \bm{V}_j \Big)  - \mbox{Tr}\Big(\bm{V}_j^H \bm{G}_{RF}^H \Big(\hat{\bm{C}}_{j} + \hat{\bm{D}}_{j} +  l_0 \bm{I} +  \bm{\Psi}_0 \Big) \bm{G}_{RF} \bm{V}_j \Big)
\end{aligned} \label{AUG_Largrangian} 
\end{equation}  \hrulefill
\end{figure*}

\section{Hybrid Beamforming and Combining} \label{sec_hybf}  

This section presents a novel HYBF design for a multi-user mmWave mMIMO FD system based on alternating optimization. In the following, optimization of the digital beamformers, analog beamformer and analog combiner is presented into separate sub-sections. We will assume the other variables to be fixed during the alternating optimization process while updating one variable. Information of the other variables updated during previous iterations will be captured in the gradients.

\subsection{Digital Beamforming}
To optimize the digital beamformers, we
take the derivative of \eqref{AUG_Largrangian} with respect to the conjugate of $\bm{U}_k$ and $\bmV_j$, which leads to the following KKT conditions
\begin{subequations} \label{grad_precoder_UL} 
\begin{equation}
\begin{aligned}
    \bm{H}_k^H & \bm{F}_{RF}\bm{R}_{\overline{k}}^{-1}  \bm{F}_{RF}^H \bm{H}_k \bm{U}_k \Big(\bm{I} + \bm{U}_k^H  \bm{H}_k^H \bm{F}_{RF} \bm{R}_{\overline{k}}^{-1}  \bm{F}_{RF}^H \\& \bm{H}_k \bm{U}_k \Big)^{-1}  -  \Big( \hat{\bm{A}}_{k} + \hat{\bm{B}}_{k}  + \bm{\Psi}_k + l_k \bm{I} \Big) \bm{U}_k = 0,
\end{aligned} 
\end{equation}  

\begin{equation}
\begin{aligned}
 & \bm{G}_{RF}^H \bm{H}_j^H \bm{R}_{\overline{j}}^{-1}  \bm{H}_j \bm{G}_{RF} \bm{V}_j \Big( \bm{I} +  \bm{V}_j^H \bm{G}_{RF}^H \bm{H}_j^H  \bm{R}_{\overline{j}}^{-1}   \bm{H}_j \bm{G}_{RF}  \\& \bm{V}_j \Big)^{-1}  -  \bm{G}_{RF}^H \Big( \hat{\bm{C}}_{j} + \hat{\bm{D}}_{j}  + \bm{\Psi}_0 + l_0 \bm{I}\Big) \bm{G}_{RF} \bm{V}_j = 0. 
\end{aligned} \label{grad_precoder_DL} 
\end{equation}
\end{subequations}

Given \eqref{grad_precoder_UL}-\eqref{grad_precoder_DL}, the digital beamformers can be optimized based on the result stated in the following.

\begin{thm}\label{optimal_digital_precoding}
Digital beamformers $\bm{U}_k$ and $\bm{V}_j$, fixed the other variables, can be optimized as the generalized dominant eigenvector solution of the pair of the following matrices 
\begin{subequations}
\begin{equation}
    \begin{aligned}
        \bm{U}_k = \bm{D}_{u_k}\Big( \bm{H}_k^H & \bm{F}_{RF} \bm{R}_{\overline{k}}^{-1}  \bm{F}_{RF}^H \bm{H}_k,\; \hat{\bm{A}}_{k} + \hat{\bm{B}}_{k}  + \bm{\Psi}_k + l_k \bm{I}\Big)
    \end{aligned} \label{gev_digital_UL_2}
\end{equation}
\begin{equation}
 \begin{aligned}
\bm{V}_j =\bm{D}_{v_j}\Big( \bm{G}_{RF}^H \bm{H}_j^H   \bm{R}_{\overline{j}}^{-1} \bm{H}_j \bm{G}_{RF}, \; \bm{G}_{RF}^H & \Big( \hat{\bm{C}}_{j} +  \hat{\bm{D}}_{j}  + \bm{\Psi}_0 \\& + l_0 \bm{I}\Big) \bm{G}_{RF} \Big),
\end{aligned} \label{gev_digital_DL}
\end{equation}
\end{subequations}
where $\bm{D}_{d}(\bm{X})$ selects $d$ generalized dominant eigenvectors from matrix $\bm{X}$.
\end{thm}
\begin{proof}
Please see Appendix \ref{hybrid_appendix}.
\end{proof}

The generalized dominant eigenvector solution provides the 
optimized beamforming directions but not power \cite{kim2011optimal}.
To include the optimal stream power allocation, we normalize the columns of digital beamformers to unit-norm. This operation preserves the optimized beamforming directions and allows to design the optimal power allocation scheme.
\vspace{-1mm}
\subsection{Analog Beamforming} \label{analog_BF_subsection}

  \begin{figure}[t] 
 \begin{subfigure}{.5\textwidth}
     	\centering
  \includegraphics[width=.75\linewidth,height=5cm]{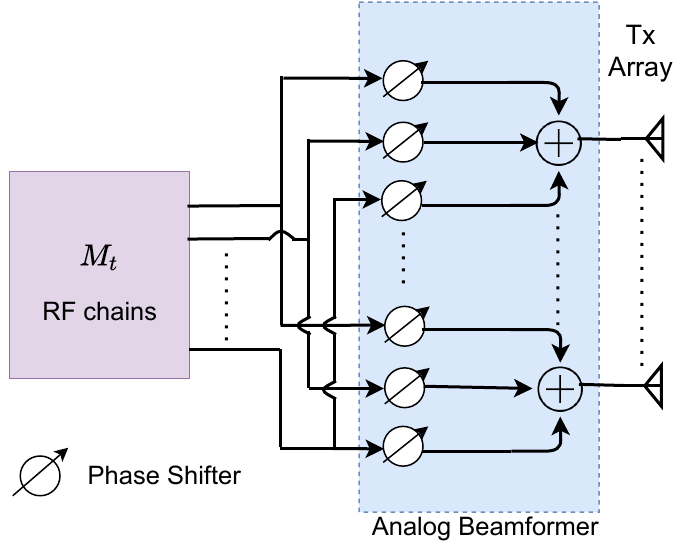}
  	\centering
 	\caption{Analog beamformer with unit-modulus phase shifters.}
 	\label{analog_BF_UM}
 \end{subfigure}
  \begin{subfigure}{.5\textwidth}
     	\centering
  \includegraphics[width=.75\linewidth,height=5cm]{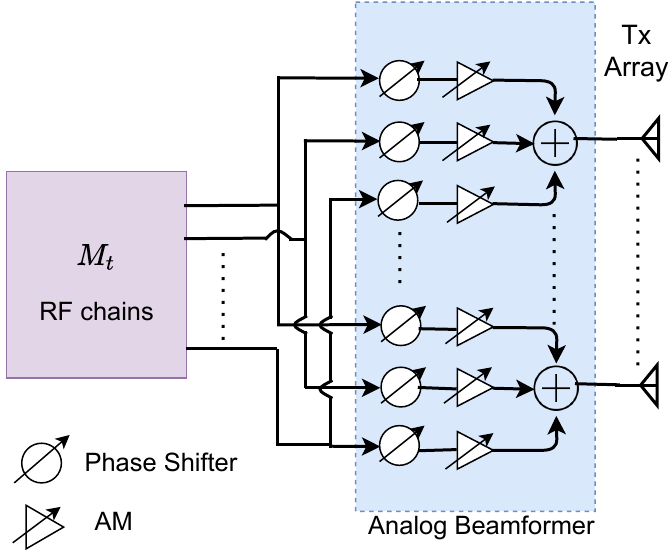}
  	\centering
 	\caption{Analog beamformer with amplitude modulators.}
 	\label{analog_BF_AM}
 \end{subfigure}  
 \caption{(a) All phase shifters are unit-modulus. (b) With amplitude control.}
 \end{figure}
 
This section presents a novel approach to design the analog beamformer for hybrid FD BS in a multi-user scenario to maximize the WSR. The structure of the fully connected analog beamformer $\bm{G}_{RF}$ is shown in Figure 2. Assuming the remaining variables to be fixed, we first consider the optimization of unconstrained analog beamformer $\bm{G}_{RF}$ as
\begin{equation}  \label{analog_restated}
\begin{aligned} 
\underset{\bm{G}_{RF}}{\text{max}.} \;\; & \sum_{j \in \mathcal{D}} w_j  \mbox{lndet}\Big( \bm{I} +  \bm{V}_j^H \bm{G}_{RF}^H \bm{H}_j^H  \bm{R}_{\overline{j}}^{-1} \bm{H}_j \bm{G}_{RF} \bm{V}_j \Big) \\& -  \mbox{Tr}\Big(\bm{V}_j^H \bm{G}_{RF}^H \Big(\hat{\bm{C}}_{j} + \hat{\bm{D}}_{j} + l_0 \bm{I} +  \bm{\Psi}_0 \Big) \bm{G}_{RF} \bm{V}_j  \Big).
\end{aligned} 
\end{equation} 

Note that from \eqref{WSR_convex} only the terms shown in \eqref{analog_restated} depend on the analog combiner $\bm{G}_{RF}$ and information about other variables is captured in gradients $\hat{\bm{C}}_{j}$ and $\hat{\bm{D}}_{j}$.
To solve \eqref{analog_restated}, we take its derivative with respect to the conjugate of $\bm{G}_{RF}$, which yields the following KKT condition
\begin{equation}
\begin{aligned}
    \bm{H}_j^H  &  \bm{R}_{\overline{j}}^{-1}  \bm{H}_j \bm{G}_{RF} \bm{V}_j \bm{V}_j^H \Big(\bm{I} + \bm{V}_j \bm{V}_j^H \bm{G}_{RF}^H \bm{H}_j^H \bm{R}_{\overline{j}}^{-1} \bm{H}_j   \\ & \bm{G}_{RF}  \Big)^{-1}   -  \Big( \hat{\bm{C}}_{j} + \hat{\bm{D}}_{j} + \bm{\Psi}_0 + l_0 \bm{I} \Big) \bm{G}_{RF} \bm{V}_j \bm{V}_j^H=0.
\end{aligned} \label{grad_analog_precoder}
\end{equation} 
Given \eqref{grad_analog_precoder}, the analog beamformer $\bm{G}_{RF}$ for mmWave FD BS can be optimized as stated in the following.

\begin{thm}\label{optimal_hybrid_precoding}
The vectorized unconstrained analog beamformer $\mbox{vec}(\bm{G}_{RF})$ can be optimized as one generalized dominant eigenvector solution of the pair of the following matrices
\begin{equation}
 \begin{aligned}
 \mbox{vec}(\bm{G}_{RF}) & =  \bm{D}_{1}\Big( \sum_{j \in \mathcal{D}}\Big(\bm{V}_j \bm{V}_j^H \Big(\bm{I} + \bm{V}_j \bm{V}_j^H \bm{G}_{RF}^H \bm{H}_j^H \bm{R}_{\overline{j}}^{-1} \\&\bm{H}_j \bm{G}_{RF} \Big)^{-1}\Big)^T  \otimes \bm{H}_j^H  \bm{R}_{\overline{j}}^{-1} \bm{H}_j,\; \\& \sum_{j \in \mathcal{D}} \Big(\bm{V}_j \bm{V}_j^H \Big)^T \otimes \Big(\hat{\bm{C}_{j}}  + \hat{\bm{D}}_{j} + \bm{\Psi}_0 + l_0 \bm{I}\Big)\Big),
    \end{aligned} \label{hybrid_AP}
\end{equation}
where $\bm{D}_{1}(\bm{X})$ selects the first generalized dominant eigenvector from matrix $\bm{X}$.
\end{thm}
\begin{proof}
Please see Appendix \ref{hybrid_appendix}.
\end{proof}
 
Note that Theorem \ref{optimal_hybrid_precoding} provides the optimized vectorized unconstrained analog beamformer $\bm{G}_{RF}$ and we need to reshape it with $\mbox{unvec}(\mbox{vec}(\bm{G}_{RF}))$. To satisfy the unit-modulus and quantization constraints, we do $\bm{G}_{RF}(m,n)=\mathbb{Q}_{P}(\angle \bm{G}_{RF}(m,n)), \forall m,n$.  For HYBF with AMs,  the columns are first scaled to be unit-norm and the quantization constraint is satisfied as $\bm{G}_{RF}(m,n) = \mathbb{Q}_{A}(|\bm{G}_{RF}(m,n)|) \mathbb{Q}_P(\angle\bm{G}_{RF}(m,n))$, $\forall m,n$.

\vspace{-1mm} 

\subsection{Analog Combining}  \label{analog_combining_subsection}
This section presents a novel approach to design the analog combiner $\bm{F}_{RF}$ for mmWave FD BS to serve multiple UL users. Its design is more straightforward than the analog beamformer. Note that the trace terms appearing in \eqref{WSR_convex} have the objective to make beamformers' update aware of the interference generated towards other links. However, $\bm{F}_{RF}$ being a combiner, does not generate any interference and therefore does not appear in the trace terms of \eqref{WSR_convex}. Consequently, to optimize $\bm{F}_{RF}$, we can solve the optimization problem \eqref{WSR} instead of using its minorized version  \eqref{WSR_convex}. By considering the unconstrained analog combiner $\bm{F}_{RF}$, from \eqref{WSR} we have
\begin{equation}
    \underset{\substack{\bm{F}_{RF}}}{\text{max}.} \quad \hspace{-2mm}\sum_{k \in \mathcal{U}} w_k \mbox{lndet}\Big( \bm{R}_{\overline{k}}^{-1} \bm{R}_k \Big). \label{WSR_combiner}
\end{equation} 
To solve \eqref{WSR_combiner}, $\bm{F}_{RF}$ has to combine the signal received at the antenna level of hybrid FD BS but $\bm{R}_k$ and $\bm{R}_{\overline{k}}$ represent the received covariance matrices after analog combining. Let ($\bm{R}_k^{ant}$) $\bm{R}_{\overline{k}}^{ant}$ denote the (signal-plus) interference and noise covariance matrix received at the antennas of FD BS, which can be obtained from ($\bm{R}_{k}$) $\bm{R}_{\overline{k}}$ given in \eqref{covariance_matrices} by omitting $\bm{F}_{RF}$. After analog combining, we can recover $\bm{R}_{k}$ and $\bm{R}_{\overline{k}}$ as $\bm{R}_k = \bm{F}_{RF}^H \bm{R}_k^{ant} \bm{F}_{RF}$ and $\bm{R}_{\overline{k}} = \bm{F}_{RF}^H \bm{R}_{\overline{k}}^{ant} \bm{F}_{RF}$, respectively, $\forall k \in \mathcal{U}$. Problem \eqref{WSR_combiner} can be restated  as a function of $\bm{R}_k^{ant}$ and $\bm{R}_{\overline{k}}^{ant}$ as

\begin{equation}
\begin{aligned}
    \underset{\substack{\bm{F}_{RF}}}{\text{max}.} \quad \hspace{-2mm}\sum_{k \in \mathcal{U}}  \Big[&w_k \mbox{lndet}\Big(\bm{F}_{RF}^H \bm{R}_k^{ant} \bm{F}_{RF} \Big) \\ & - w_k \mbox{lndet}\Big( \bm{F}_{RF}^H \bm{R}_{\overline{k}}^{ant}\bm{F}_{RF} \Big) \Big].
\end{aligned} \label{WSR_combiner_restated}
\end{equation} 

In \eqref{WSR_convex}, the trace term was only linear, which made the restated optimization problem concave for each link. In \eqref{WSR_combiner_restated}, all the terms are fully concave. To optimize $\bm{F}_{RF}$, we take the derivative with respect to the conjugate of $\bm{F}_{RF}$, which yields the following KKT condition
\begin{equation} \label{kkt_combiner}
\begin{aligned}
\sum_{k \in \mathcal{U}} w_k & \bm{R}_k^{ant} \bm{F}_{RF}  \Big( \bm{F}_{RF}^H \bm{R}_k^{ant} \bm{F}_{RF} \Big)^{-1} \\& - \sum_{k \in \mathcal{U}} w_k\bm{R}_{\overline{k}}^{ant} \bm{F}_{RF} \Big( \bm{F}_{RF}^H \bm{R}_{\overline{k}}^{ant} \bm{F}_{RF} \Big)^{-1} = 0.
\end{aligned}
\end{equation}
It is immediate from \eqref{kkt_combiner} that the unconstrained analog combiner can be optimized as the generalized dominant eigenvector solution of the pair of sum of the received covariance matrices at the antenna level from all the $K$ UL users, i.e.,
\vspace{-1mm}
\begin{equation}
    \begin{aligned}
        \bm{F}_{RF} \rightarrow \bm{D}_{Nr}\Big( \sum_{k \in \mathcal{U}} w_k\bm{R}_k^{ant},\;\sum_{k \in \mathcal{U}} w_k\bm{R}_{\overline{k}}^{ant}\Big).
    \end{aligned} \label{gev_combiner}
\end{equation}
To satisfy the unit-modulus and quantization constraints for $\bm{F}_{RF}$, we do $\bm{F}_{RF}(m,n)= \mathbb{Q}_P(\angle \bm{F}_{RF}(m,n)) \in \mathcal{P}$, $\forall m,n $. If AMs are available, the columns are scaled to be unit-norm and quantization constraint is satisfied as $\bm{F}_{RF}(m,n)= \mathbb{Q}_{A}(|\bm{F}_{RF}(m,n))| \mathbb{Q}_P(\angle \bm{F}_{RF}(m,n)), \forall m,n$.

\subsection{Optimal Power Allocation}
Given the normalized digital beamformers and analog beamformer, optimal power allocation can be included while searching for the Lagrange multipliers satisfying the joint sum-power and practical per-antenna power constraints.

Let $\bm{\Sigma}_k^{(1)}$ and $\bm{\Sigma}_k^{(2)}$, $\forall k \in \mathcal{U}$ and $\bm{\Sigma}_j^{(1)}$ and $\bm{\Sigma}_j^{(2)}$, $\forall j \in \mathcal{D}$, be defined as
\begin{subequations}
\begin{equation}
   \bm{U}_k^H  \bm{H}_k^H \bm{F}_{RF} \bm{R}_{\overline{k}}^{-1}  \bm{F}_{RF}^H \bm{H}_k \bm{U}_k   = \bm{\Sigma}_k^{(1)}, \label{diag_1}
\end{equation}
\begin{equation}
    \bm{U}_k^H \Big(\hat{\bm{A}}_{k} + \hat{\bm{B}}_{k} + \bm{\Psi}_k + l_k  \bm{I}\Big) \bm{U}_k = \bm{\Sigma}_k^{(2)},\label{diag_2}
\end{equation}
\begin{equation}
     \bm{V}_j^H \bm{G}_{RF}^H \bm{H}_j^H \bm{R}_{\overline{j}}^{-1} \bm{H}_j \bm{G}_{RF} \bm{V}_j = \bm{\Sigma}_j^{(1)},\label{diag_3}
\end{equation}
\begin{equation}
    \bm{V}_j^H \bm{G}_{RF}^H \Big(\hat{\bm{C}}_{j} +  \hat{\bm{D}}_{j} + \bm{\Psi}_0  + l_0 \bm{I}\Big) \bm{G}_{RF} \bm{V}_j = \bm{\Sigma}_j^{(2)}. \label{diag_4}
\end{equation}\label{diagonal_matrices}
\end{subequations}
Given \eqref{diagonal_matrices}, the optimal stream power allocation can be included based on the result stated in the following.

\begin{lemma} \label{power_allocation}
Optimal power allocation for the hybrid FD BS and multi-antenna UL users can be obtained by multiplying $\bm{\Sigma}_j^{(1)}$ and $\bm{\Sigma}_j^{(2)}$ with the diagonal power matrix $\bm{P}_j$, $\forall j \in \mathcal{D}$ and $\bm{\Sigma}_k^{(1)}$ and $\bm{\Sigma}_k^{(2)}$ with the diagonal power matrix $\bm{P}_k$, $\forall k \in \mathcal{U}$, respectively. 
\end{lemma}
\begin{proof} The beamformers $\bm{U}_k$ and $\bm{V}_k$, are computed as the generalized dominant eigenvectors, which make the matrices $\bm{\Sigma}_k^{(1)},\bm{\Sigma}_k^{(2)}, \forall k$ and $\bm{\Sigma}_j^{(1)},\bm{\Sigma}_j^{(2)}, \forall j$ diagonal at each iteration.
Multiplying any generalized dominant eigenvector solution matrix with a diagonal matrix still yields a generalized dominant eigenvector solution. Therefore, multiplying  $\bm{\Sigma}_k^{(1)},\bm{\Sigma}_k^{(2)}$ with $\bm{P}_k$, $\forall k \in \mathcal{U}$ and $\bm{\Sigma}_j^{(1)},\bm{\Sigma}_j^{(2)}$ with $\bm{P}_j, \forall j \in \mathcal{D}$ still preserves the validity of optimized beamforming directions.
\end{proof}

Given the optimized beamformers and fixed Lagrange multipliers, by using the result stated in Lemma \ref{power_allocation}, stream power allocation optimization problems for  UL and DL users can be formally stated as 
\begin{subequations} \label{power_opt}
\begin{equation} 
\centering
\begin{aligned}
& \underset{\bm{P}_k}{\text{max.}}\; w_k \mbox{lndet}\Big( \bm{I} + \bm{\Sigma}_k^{(1)} \bm{P}_k \Big) - \mbox{Tr}\Big(\bm{\Sigma}_k^{(2)} \bm{P}_k \Big),\quad \forall k \in \mathcal{U}, \\
\end{aligned} \label{UL_pow}
\end{equation}
\begin{equation}
\centering
\begin{aligned}
& \underset{\bm{P}_j}{\text{max.}} w_j \mbox{lndet}\Big( \bm{I} +  \bm{\Sigma}_j^{(1)}\bm{P}_j \Big) - \mbox{Tr}\Big( \bm{\Sigma}_j^{(2)}  \bm{P}_j \Big),\quad \forall j \in \mathcal{D}. \\
\end{aligned}\label{DL_pow}
\end{equation}
\end{subequations}

Solving \eqref{power_opt} leads to the following optimal power allocation scheme
  
\begin{subequations} \label{optimal_pow}
 \begin{equation}
 \begin{aligned}
      \bm{P}_k = \Big( w_k  & \Big( \bm{U}_k^H \Big(\hat{\bm{A}}_{k} + \hat{\bm{B}}_{k} + \bm{\Psi}_k + l_k \bm{I}\Big) \bm{U}_k \Big)^{-1} \\& - \Big(\bm{U}_k^H  \bm{H}_k^H \bm{F}_{RF}  \bm{R}_{\overline{k}}^{-1}  \bm{F}_{RF}^H \bm{H}_k \bm{U}_k \Big)^{-1} \Big)^{+}, \label{UL_powers}
 \end{aligned}
 \end{equation}
 \begin{equation}
 \begin{aligned}
     \bm{P}_j = \Big( w_j  & \Big( \bm{V}_j^H \bm{G}_{RF}^H\Big(\hat{\bm{C}}_{j} +  \hat{\bm{D}}_{j} + \bm{\Psi}_0  + l_0 \bm{I}\Big)\bm{G}_{RF}\bm{V}_j \Big)^{-1}\\& - \Big(\bm{V}_j^H \bm{G}_{RF}^H \bm{H}_j^H \bm{R}_{\overline{j}}^{-1} \bm{H}_j \bm{G}_{RF} \bm{V}_j\Big)^{-1} \Big)^{+}, \label{DL_powers}
 \end{aligned}
 \end{equation}
\end{subequations}  

where $(\bm{X})^+ = \mbox{max}\{\bm{0},\bm{X}\}$. We remark that the proposed power allocation scheme is interference, SI, cross-interference and LDR noise aware as it takes into account their effect in the gradients, which are updated at each iteration. Fixed the beamformers, we can search for multipliers satisfying the joint constraints while doing water-filling for powers. To do so, consider the dependence of Lagrangian \eqref{AUG_Largrangian} on multipliers and powers as
\vspace{-1mm}
\begin{equation} 
\begin{aligned}
& \mathcal{L}(\bm{\Psi},\bm{L}, \bm{P})  = \sum_{l=0}^{K} l_l p_l + \mbox{Tr}\Big(\bm{\Psi_0 \bm{\Lambda_0}}\Big) + \sum_{u \in \mathcal{U}} \mbox{Tr}\Big(\bm{\Psi_u \bm{\Lambda_u}}\Big) \\
 &+ \sum_{k \in \mathcal{U}} w_k \mbox{lndet}\Big( \bm{I} + \bm{\Sigma}_k^{(1)} \bm{P}_k \Big) - \mbox{Tr}\Big(\bm{\Sigma}_k^{(2)} \bm{P}_k \Big)  \\
 & + \sum_{j\in \mathcal{D}} w_j \mbox{lndet}\Big( \bm{I} + \bm{\Sigma}_j^{(1)} \bm{P}_j\Big) - \mbox{Tr}\Big(\bm{\Sigma}_j^{(2)} \bm{P}_j \Big) ,
\end{aligned} \label{Largrangian} 
\end{equation}
where $\bm{P}$ is the set of stream powers in UL and DL. The multipliers in $\bm{\Psi}$ and $\bm{L}$ should be such that the Lagrange dual function \eqref{Largrangian} is finite and the values of multipliers should be strictly positive. Formally, Lagrange multipliers' search problem can be stated as
\begin{equation} \label{dual_func}
\begin{aligned}
  & \underset{\bm{\Psi},\bm{L}}{\text{min}.}\,\, \underset{\bm{P}}{\text{max}.}   \quad  & \mathcal{L}\Big(\bm{\Psi},\bm{L},\bm{P}\Big), \\
  &\quad \quad \mbox{s.t.}  & \bm{\Psi} ,\bm{L} \succeq 0.
\end{aligned}
\end{equation}
The dual function $\underset{\bm{P}}{\text{max}.}\;\mathcal{L}(\bm{\Psi},\bm{L},\bm{P})$ is the pointwise supremum of a family of functions of $\bm{\Psi},\bm{L}$, it is convex \cite{boyd2004convex} and the globally optimal values for $\bm{\Psi}$ and $\bm{L}$ can be obtained by using any of the numerous convex optimization techniques. In this work, we adopt the Bisection algorithm to search the multipliers. 
Let $\mathcal{M}_0 = \{\lambda_0,\psi_1,..,\psi_{M0}\}$ and $\mathcal{M}_k = \{\lambda_k,\psi_{k,1},..,\psi_{k,M_k}\}$ denote the sets containing Lagrange multipliers associated with the sum-power and practical per-antenna power constraints for FD BS and UL user $k \in \mathcal{U}$, respectively. Let $\underline{\mu_i}$ and $\overline{\mu_i}$ denote the lower and upper bound for the search range of multiplier $\mu_i$, where $\mu_i \in \mathcal{M}_0$ or $ \mu_i \in \mathcal{M}_k$. 
While searching multipliers and performing water-filling for powers, the UL and DL power matrices become non-diagonal. Therefore, we consider the SVD of power matrices to shape them back as diagonal. Namely, let $\bm{P}_i$ denote the power matrix for user $i$, where $i \in \mathcal{U} \;\mbox{or} \;i \in\mathcal{D}$. When $\bm{P}_i$ becomes non-diagonal, we consider its SVD as
  \begin{equation}
     [\bm{U}_{P_i},\bm{D}_{P_i},\bm{V}_{P_i}] = \mbox{SVD}(\bm{P}_i).
     \label{eq_SVD_P}
 \end{equation}
 where $\bm{U}_{P_i},\bm{D}_{P_i}$ and $\bm{V}_{P_i}$ are the left unitary, diagonal and right unitary matrices, respectively, obtained with the SVD decomposition, and we set $\bm{P}_i = \bm{D}_{P_i}$ to obtain diagonal power matrices.
 
 For unit-modulus HYBF, the complete alternating optimization based procedure to maximize the WSR based on minorization-maximization is formally stated in Algorithm $1$. For HYBF with AMs, the steps $\angle \bm{G}_{RF}$ and $\angle \bm{F}_{RF}$ must be omitted and amplitudes of the analog beamformer and combiner must be quantized with $\mathbb{Q}_{A}(\cdot)$. Once the proposed algorithm converges, all the combiners can be chosen as the MMSE combiners, which will not affect the WSR achieved with Algorithm 1 $(4)-(9)$ \cite{christensen2008weighted}.

\begin{algorithm}[!t] 
\caption{Practical Hybrid Beamforming Design }\label{alg_1}
\textbf{Given:} $\mbox{The CSI and rate weights.}$\\
\textbf{Initialize:}\;$\bm{G}_{RF}, \bm{V}_j, \forall j \in \mathcal{D}$ and $\bm{U}_k, \forall k \in \mathcal{U} $.\\
\textbf{Set:} $\underline{\mu_i} = 0$ and $\overline{\mu_i} =\mu_{i_{max}}$ $\forall i \in \mathcal{M}_0$ or $\forall i \in \mathcal{M}_k$, $\forall k \in \mathcal{U}$ \\
\textbf{Repeat until convergence}
\begin{algorithmic}
\STATE \hspace{0.1cm} Compute $\bm{G}_{RF}$ \eqref{hybrid_AP}, $ \mbox{unvec}(\bm{G}_{RF})$ and $\bm{G}_{RF} =  \angle \bm{G}_{RF}$.\\
\STATE \hspace{0.1cm} Compute $\bm{F}_{RF}$ with \eqref{gev_combiner}, and do $\bm{F}_{RF} =  \angle \bm{F}_{RF}$.\\
\STATE \hspace{0.1cm} \textbf{for:} $j =1:J$ 
\STATE \hspace{0.5cm} Compute $\hat{\bm{C}}_{j}, \hat{\bm{D}}_{j}$ with \eqref{gradients}\\
\STATE \hspace{0.5cm} Compute $\bm{V}_j$ with \eqref{gev_digital_DL} and normalize it\\
\STATE \hspace{0.1cm} \textbf{end}
\STATE \hspace{0.5cm} \textbf{Set:} $\underline{\mu_0} = 0$ and $\overline{\mu_0} =\mu_{i_{max}}$ $\forall i \in \mathcal{M}_0$  \\
\STATE \hspace{0.5cm} \textbf{for:} $\forall \mu_0 \in \mathcal{M}_0$
\STATE \hspace{0.9cm} \textbf{Repeat until convergence}\\
\STATE \hspace{1.1cm} set $\mu_0 = (\underline{\mu_0} + \overline{\mu_0})/2$  \\ 
\STATE \hspace{1.1cm} Compute $\bm{P}_j$ with \eqref{DL_powers} $\forall j$ \\
\STATE \hspace{1.1cm} \textbf{if} constraint for $\mu_0$ is violated\\
\STATE \hspace{1.3cm} set $\underline{\mu_0} = \mu_0$, \\
\STATE \hspace{1.1cm} \textbf{else} $\overline{\mu_0} = \mu_0$\\
\STATE \hspace{0.3cm}$[\bm{U}_{P_j},\bm{D}_{P_j},\bm{V}_{P_j}] = \mbox{SVD}(\bm{P}_j), \forall j$\\
\STATE \hspace{0.3cm} Set $\bm{P}_j = \bm{D}_{P_j}$ and $\bm{Q}_j = \bm{G}_{RF} \bm{V}_j \bm{P}_j \bm{V}_j^H \bm{G}_{RF}^H, \forall j$
\STATE \hspace{0.1cm} \textbf{for:} $ k = 1:K$  \\
\STATE \hspace{0.5cm} Compute $\hat{\bm{A}}_{k}, \hat{\bm{B}}_{k}$ with \eqref{gradients}\\
\STATE \hspace{0.5cm} Compute $\bm{U}_k$ with \eqref{gev_digital_UL_2} and normalize it\\
\STATE \hspace{0.5cm} \textbf{Set:} $\underline{\mu_k} = 0$ and $\overline{\mu_k} =\mu_{l_{max}}$ \\
\STATE \hspace{0.5cm} \textbf{for:} $\forall \mu_k \in \mathcal{M}_k$
\STATE \hspace{0.9cm} \textbf{Repeat until convergence}\\
\STATE \hspace{1.1cm} set $\mu_k = (\underline{\mu_k} + \overline{\mu_k})/2$  \\ 
\STATE \hspace{1.1cm} Compute $\bm{P}_k$ with \eqref{UL_powers}. 
\STATE \hspace{1.1cm} \textbf{if} constraint for $\mu_0$ is violated\\
\STATE \hspace{1.3cm} set $\underline{\mu_k} = \mu_k$ \\
\STATE \hspace{1.1cm} \textbf{else} $\overline{\mu_k} = \mu_k$\\
\STATE \hspace{0.5cm} $[\bm{U}_{P_k},\bm{D}_{P_k},\bm{V}_{P_k}] = \mbox{SVD}(\bm{P}_k)$\\
\STATE \hspace{0.5cm} Set $\bm{P}_k = \bm{D}_{P_k}$ and $\bm{T}_k = \bm{U}_k \bm{P}_k \bm{U}_k^H$
\STATE Repeat 
\STATE  Quantize $\angle \bm{G}_{RF}$ and $\angle \bm{F}_{RF}$ ($| \bm{G}_{RF}|$ and $|\bm{F}_{RF}|$ with AMs) \\
\end{algorithmic}
\label{algo1} 
\end{algorithm}
 \vspace{-1.5mm}
\subsection{Convergence}
 
In our context, the ingredients required to prove the convergence are minorization \cite{StoicaSelen:SPmag0104}, alternating or cyclic optimization \cite{StoicaSelen:SPmag0104}, Lagrange dual function
\cite{boyd2004convex}, saddle-point interpretation \cite{boyd2004convex} and KKT conditions \cite{boyd2004convex}. For the WSR cost function \eqref{WSR}, we construct its minorizer as in \eqref{eqAk}, \eqref{eqBk}, \eqref{eqCk}, \eqref{eqDk},
which restates the WSR maximization as a concave problem \eqref{WSR_convex} for each link. The minorizer is a touching lower bound for the original WSR problem \eqref{WSR}, so we can write

\begin{equation}
\begin{aligned}
\mbox{WSR} & \geq   \underline{\mbox{WSR}}  = \underline{\mbox{WR}}_{k}^{UL} + \underline{\mbox{WSR}}_{\overline{k}}^{UL} + \underline{\mbox{WR}}_{j}^{DL} +\underline{\mbox{WSR}}_{\overline{j}}^{DL}.
\end{aligned}
\label{eqWSR2}
\end{equation}
The minorizer, which is concave in $\bm{T}_k$ and $\bm{Q}_j$, still has the same gradient of the original WSR
and hence the KKT conditions are not affected. Reparameterizing $\bm{T}_k$ or $\bm{Q}_j$ in terms of $\bm{U}_k, \forall k \in \mathcal{U}$ and $\bm{G}_{RF}\; \mbox{or}\; \bm{V}_j, \forall j \in \mathcal{D}$, respectively, as in \eqref{reparametarization}  with the optimal power matrices and adding the power constraints to the minorizer, we get the Lagrangian
\eqref{AUG_Largrangian}. Every alternating update of $\mathcal{L}$ for $\bm{V}_j$, $\bm{G}_{RF}$, $\bm{U}_k,\forall j \in \mathcal{D}, \forall k \in \mathcal{U}$ or for $\bm{P},\bm{\Lambda},\bm{\Psi}$ leads to an increase of the WSR,
ensuring convergence. For the KKT conditions, at the convergence point, the gradients of $\mathcal{L}$ for $\bm{V}_j$,$\bm{G}_{RF}$, $\bm{U}_j$ or $\bmP$ correspond to the gradients of Lagrangian \eqref{WSR}, i.e., for the original WSR problem. For fixed analog and digital beamformers, $\mathcal{L}$ is concave in $\bmP$, hence we have a strong duality for the saddle point, i.e.
\begin{equation}
    \max_{\bmP} \min_{\bm{L},\bPsi}. \mathcal{L}\Big(\bm{L},\bPsi,\bmP\Big). 
\end{equation}
Let $\bm{X}^*$ and $x^*$
denote the optimal solution for matrix $\bm{X}$ or scalar $x$ at the convergence, respectively. When Algorithm 1 converges, solution of the following optimization problem 
\begin{equation}
    \min_{\bLambda,\bPsi} \mathcal{L}\Big(\bmV^*,\bmG^{*},\bmU^*, \bmP^*,\bm{L},\bm{\Psi}\Big)
\end{equation}
  satisfies the KKT conditions for powers in $\bmP$ and the complementary slackness conditions
\begin{subequations}
\begin{equation}
\begin{aligned}
l_0^* \, \Big(\alpha_0 - \sum_{j\in \mathcal{D}} \mbox{Tr}\Big( \bm{G}_{RF}^{*}\bm{V}_j^{*} \bmP_j^*\bmV_j^{*\,  H}\bm{G}_{RF}^{*\,H}\Big)\Big) &= 0,
\end{aligned}
\end{equation}
\begin{equation}
\begin{aligned}
\mbox{Tr}\Big(\bm{\Psi}_0^* \, \Big(\bmP_0 - \sum_{j\in \mathcal{D}} \mbox{Tr}\Big( \bm{G}_{RF}^{ *}\bmV_j^{*} \bmP_j^*\bmV_j^{*\,  H}\bm{G}_{RF}^{*\,H}\Big)\Big)\Big) &= 0, 
\end{aligned}
\end{equation}
\begin{equation}
\begin{aligned}
l_k^* \, \Big(\alpha_k -  \mbox{Tr}\Big( \bmU_k^{*} \bmP_k^*\bmU_k^{*\,  H}\Big)\Big) & = 0, 
\end{aligned}
\end{equation}
\begin{equation}
\begin{aligned}
\mbox{Tr}\Big(\bm{\Psi}_k^* \, \Big(\bmP_k -  \mbox{Tr}\Big( \bmU_k^{*} \bmP_k^*\bmU_k^{ *\,  H}\Big)\Big)\Big) &= 0,
\end{aligned}
\end{equation}
\label{eqslackness}
\end{subequations}
where all the individual factors in the products are non-negative, and for  per-antenna power constraints $\bm{\Psi}_0^*$ and $\bm{\Psi}_k^*$, the sum of non-negative terms being zero implies all terms result to be zero. 

\begin{figure}
    \centering
    \includegraphics[width=0.49\textwidth]{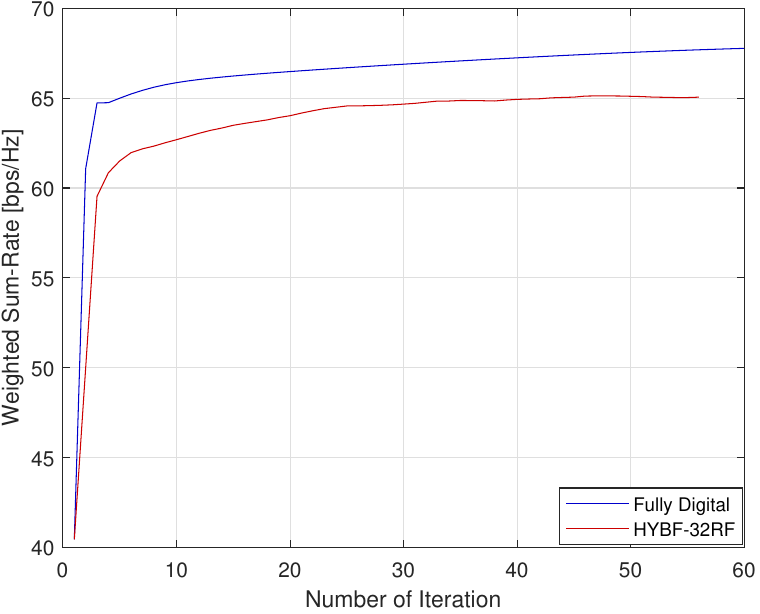}
    \caption{Typical convergence behaviour of the proposed HYBF for mmWave mMIMO FD system.}
    \label{convergenza_alg}
\end{figure}

\vspace{-0.5mm}
\emph{Remark 3:} The unit-modulus HYBF scheme converges to a local optimum where $\angle \bmG_{RF}(m,n), \angle \mathbf{F}_{RF}(m,n)\hspace{-2mm} \in \mathcal{P}$ with $|\bmG_{RF}(m,n)|,|\mathbf{F}_{RF}(m,n)|\hspace{-2mm}=\hspace{-2mm}1, \forall m,n$. Unconstrained HYBF with AMs converges to a different local optimum, where $\angle \bmG_{RF}(m,n), \angle \mathbf{F}_{RF}(m,n) \in \mathcal{P}$ and $|\bmG_{RF}(m,n)|,\; |\mathbf{F}_{RF}(m,n)| \in \mathcal{A},\; \forall m,n$. Due to quantization, $\bmG_{RF}$ and $\mathbf{F}_{RF}$ obtained with Algorithm $1$ tend to lose their optimality and consequently achieve less WSR compared to their infinite resolution case. For unit-modulus HYBF, the loss in WSR depends only on the resolution of phases. For HYBF with AMs, the loss in WSR depends on the resolution of both amplitudes and phases.

\subsection{Complexity Analysis} \label{complexity_analysis}
In this section, we analyze the per-iteration computational complexity of Algorithm 1, assuming that the dimensions of antennas get large. Its one iteration consists in updating $K$ and $J$ digital beamformers for the UL and DL user, respectively, and one analog beamformer and combiner for the FD BS. One dominant generalized eigenvector computation to update analog beamformer $\bm{G}_{RF}$ from a matrix of size $M_t M_0 \times M_t M_0$ in \eqref{hybrid_AP}, is $\mathcal{O}\left(M_0^2 M_t^2\right)$. To update the gradients $\hat{\bmA}_k$ and $\hat{\bmB}_k$ for one UL user, the complexity is given by $\mathcal{O}((K-1) N_r^3)$ and $\mathcal{O}(J N_j^3)$, respectively. For the gradient $\hat{\bmC}_j$ and $\hat{\bmD}_j$, required to update the beamformer of $j$-th DL user, computational complexity is $\mathcal{O}((J-1) N_j^3)$ and $\mathcal{O}(K N_r^3)$, respectively. Updating the beamformers of $k$-th UL and $j$-th DL users as the generalized dominant eigenvectors adds additional complexity of $\mathcal{O}(u_k M_k^2)$ and $\mathcal{O}(v_j N_j^2)$, respectively. The Lagrange multipliers' update associated with the per-antenna power constraints for FD BS and UL users is linear in the number of antennas $M_0$ or $M_k$, respectively. However, as we jointly perform the multipliers' search and power allocation, it adds $\mathcal{O}(v_i^3),$ where $i \in \mathcal{D}$ or $i \in \mathcal{U}$, which can be ignored. Updating the analog combiner $\bm{F}_{RF}$ for FD BS is $\mathcal{O}(N_r N_0^2)$. Under the assumption that the dimensions of antennas get large, the per-iteration complexity is
$\approx  \mathcal{O}( K^2 N_r^3 + K J N_j^3 + J^2 N_j^3 + J K N_r^3 + M_0^2 M_t^2 + N_r N_0^2 )$ which depends on the number of UL and DL users served by the mmWave FD BS.

\section{Simulation Results} \label{simulazioni}
This section presents simulation results to evaluate the performance of the proposed HYBF scheme. For comparison, we define the following benchmark schemes:\\
 
a)  A \emph{Fully digital HD} scheme with LDR noise, serving the UL and  DL users with time division duplexing. Being HD, it is neither affected by the SI nor by the cross-interference.
\\

b)  A \emph{Fully digital FD} scheme with LDR noise. This scheme sets an upper bound for the maximum achievable gain by a hybrid FD system.
\\
  
Hereafter, HYBF designs with the unit-modulus constraint and with AMs are denoted as HYBF-UM and HYBF-AMs, respectively. We define the signal-to-noise-ratio (SNR) for the mmWave mMIMO FD system as 
\begin{equation} 
    \mbox{SNR} = \alpha_0/\sigma_0^2,
\end{equation}

where the scalars $\alpha_0$ and $\sigma_0^2$ denote the total transmit power and thermal noise variance for FD BS, respectively. We set the thermal noise level for DL users to be $\sigma_0^2 = \sigma_j^2, \forall j$, and the transmit power for UL users as $\alpha_0=\alpha_k$, $\forall k$. We consider the total transmit power normalized to $1$ and choose the noise variance based on desired SNR. To compare the gain of a FD system over a HD system, we define the additional gain in percentage as 
\begin{equation}
\mbox{Gain} = \frac{WSR_{FD} - WSR_{HD}}{WSR_{HD}}  \times 100 \;[\%],
\end{equation}

where $WSR_{FD}$ and $WSR_{HD}$ denote the WSR of a FD and HD system, respectively.
To evaluate the performance, we set the per-antenna power constraints for FD BS and UL users as the total transmit power divided by the number of antennas, i.e. $\alpha_0/M_0 \bmI\; \mbox{and}\; \alpha_k/M_k \bmI, \forall k.$ The BS and users are assumed to be equipped with a uniform linear array (ULA) with antennas separated by half-wavelength. The transmit and receive antenna array at the BS are assumed to be placed $D = 20$~cm apart, with the relative angle $\Theta = 90^\circ$, and $r_{m,n}$ is modelled as (9) \cite{satyanarayana2018hybrid}. The Rician factor $\kappa$ for the SI channel is set to be $1$.
We assume that the FD BS has $M_0 = 100$ transmit and $N_0=50$ receive antennas. It serves two UL and two DL users with $M_k = N_j = 5$ antennas and with $2$ data streams for each user. The phases for both designs are quantized in the interval $[0,2\pi]$ with an $8$-bit uniform quantizer $\mathbb{Q}_{P}(\cdot)$. For HYBF with AMs, the amplitudes are uniformly quantized with a $3$-bit uniform quantizer $\mathbb{Q}_{A}(\cdot)$ in the interval $[0, a_{max}]$, where $a_{max} = \mbox{max}\{| \mbox{max}\{\mathbf{G}_{RF}\}|,\mbox{max}\{|\mathbf{F}_{RF}|\}\}$ is the maximum of the maximum modulus of $\mathbf{G}_{RF}$ or $\mathbf{F}_{RF}$. We assume the same LDR noise level for the users and FD BS, i.e. $k_0 = \beta_0 = \kappa_k = \beta_j$. The rate weights for the UL and DL users are set to be $1$.
Aforementioned simulation parameters are summarized in Table \ref{table_parametri}. The digital beamformers are initialized as the dominant eigenvectors of the channel covariance matrices of the intended users. Analog beamformer and combiner are initialized as the dominant eigenvectors of the sum of channel covariance matrices across all the UL and DL users, respectively. Note that as we assume perfect CSI, the SI can be cancelled with HYBF only up to the LDR noise level, which represents the residual SI.

 \begin{table}[!t]
\centering
    \caption{Simulation parameters to simulate the multi-user mmWave FD system.}
    \resizebox{8cm}{!}{%
    \begin{tabular}{|p{29mm}|p{20mm}|p{26mm}|}
        \hline
       \multicolumn{3}{|c|}{Simulation Parameters} \\
       \hline
      UL and DL users &\mbox{$ K,J$}  & 2 \\ \hline
      Data streams   &\mbox{$v_j$},\mbox{$u_k$} & 2   \\ \hline
       Antennas for the BS  &\mbox{$M_0, N_0$} &  100, 50  \\ \hline
       Clusters and Paths  &\mbox{$N_c$},\mbox{$N_p$} & 3,3  \\ \hline
       RF chains (BS)  &\mbox{$ M_t = N_r$} & 8,10,16 or 32  \\  \hline
       User antennas &\mbox{$M_k = N_j$}  &  5  \\  \hline
      Rician Factor   &\mbox{$\kappa$} & 1 \\ \hline
      Tx and Rx array response   &\mbox{$\bm{a}_r$},\mbox{$\bm{a}_t$} &  ULA,ULA \\\hline
       Angles  &\mbox{$\phi_k$},\mbox{$\phi_j$},\mbox{$\theta_k$},\mbox{$\theta_j$} &  $\mathcal{U}$\mbox{$\sim [-30^{\circ},30^{\circ}]$} \\ \hline
        Rate weights &$w_k,w_j$ & 1\\\hline
       Uniform Quantizer &$\mathbb{Q}_P(\cdot), \mathbb{Q}_{A}(\cdot)$ &  $8,3$ bits\\\hline
       Angle between Tx and Rx array (BS) &$\Theta$ & $90^{\circ}$\\\hline
     Antenna array separation (BS) &  $D$ & 20 cm\\\hline
      Per-antenna power constraint  &  $\bm{\Lambda}_k$,$\bm{\Lambda}_0$ & $\alpha_k/M_k \bm{I}$,$\alpha_0/M_0\;\bm{I}$\\ \hline
    \end{tabular}}
          \label{table_parametri}
\end{table}

Figure \ref{fig_100ant_0dB} shows the achieved average WSR with the proposed HYBF designs as a function of the LDR noise with $\mbox{SNR}=0~$dB. The fully digital FD scheme achieves an additional gain of $\sim 97\%$ over a fully digital HD scheme. The impact of different LDR noise levels on the maximum achievable WSR for a mmWave FD system with different number of RF chains is also shown. For $k_0 \leq -40~$dB,  HYBF-UM and HYBF-AMs achieve an additional gain of $\sim 85 \%, 64 \%, 42\%, 3 \%$  and $\sim 89 \%, 74 \%, 60\%, 28 \%$ with $32,16,10,8$ RF chains, respectively. We can see that as the LDR noise variance increases, achievable WSR for both the hybrid FD and fully digital HD system degrades severely. Figure \ref{fig_100ant_40dB} shows the achieved average WSR as a function of the LDR noise with $\mbox{SNR}=40$dB. For $k_0 \leq -80~$dB, HYBF-UM and HYBF-AMs achieve an additional gain of $\sim65 \%, 55\%, 41\%, 15\%$ and $\sim 67 \%, 62 \%, 55 \%, 26 \%$ with $32,16,10,8$ RF chains, respectively, and increasing the LDR noise variance degrades the achieved average WSR. By comparing Figure \ref{fig_100ant_0dB} with Figure \ref{fig_100ant_40dB}, we can see that at low SNR, HYBF-UM with only $8~$RF chains performs close to the fully digital HD scheme. As the SNR increases to $40~$dB, HYBF-UM with $8~$RF achieves an additional gain of $\sim 15\%$. HYBF-AMs with only $8$ RF chains outperforms the fully digital HD scheme for all the SNR levels. Figures \ref{fig_100ant_0dB}-\ref{fig_100ant_40dB} also show that HYBF-AMs with $10$ RF chains achieves similar average WSR as the HYBF-UM with $16$ RF chains. It is interesting to observe that increasing the SNR from $0~$dB to $40~$dB decreases the thermal noise variance and the LDR noise variance dominates the noise floor already with $k_0 = -80~$dB at SNR$=40~$dB. For SNR$=0~$dB, the LDR noise variance dominates only for $k_0 > -40~$dB. From this observation, we can conclude that hardware with a low LDR noise is required to benefit from a high SNR in the mmWave FD systems.
\begin{figure}[t]
     \centering
   \includegraphics[width=0.49\textwidth]{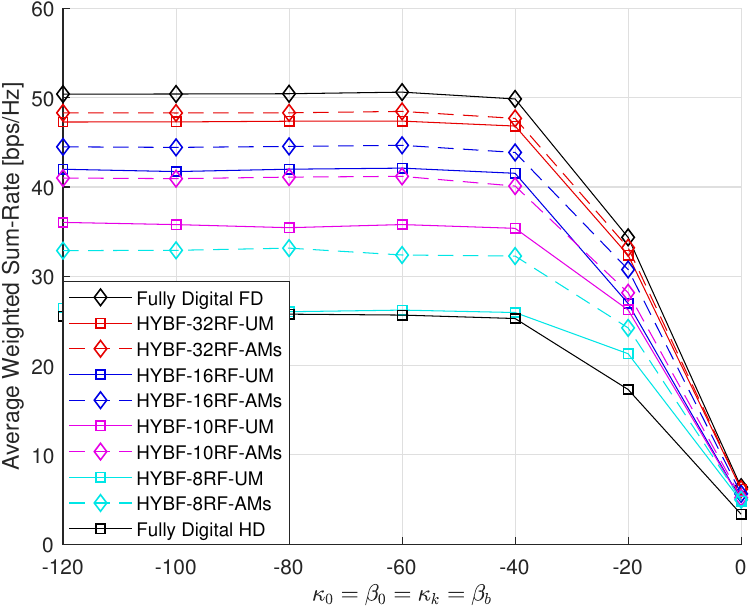}
     \caption{Average WSR as a function of the LDR noise with SNR = $0$~dB.}
       \label{fig_100ant_0dB}
\end{figure}

\begin{figure}[t]
     \centering
       \includegraphics[width=0.49\textwidth]{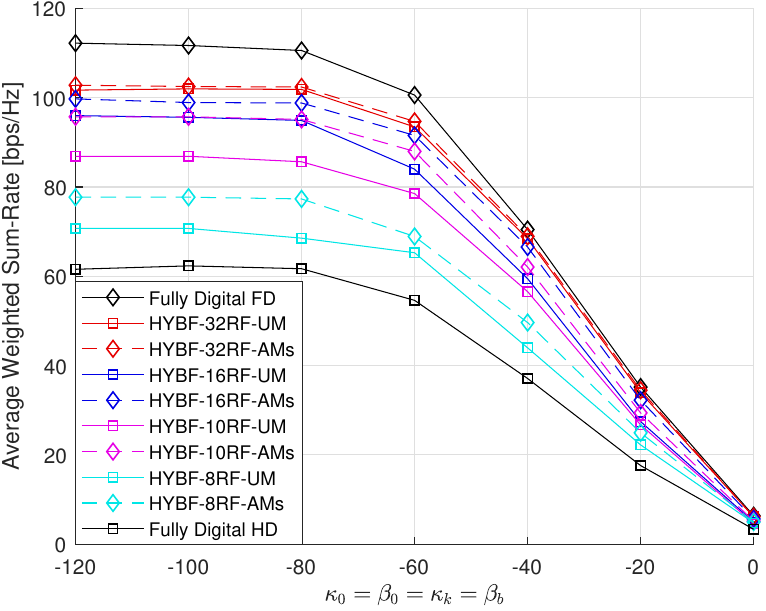}  
        \caption{Average WSR as a function of the LDR noise with SNR = $40$~dB.}
      \label{fig_100ant_40dB}
\end{figure}

 \begin{figure*}[!t]
    \centering
    \begin{minipage}{0.49\textwidth}
        \centering
    \includegraphics[width=\textwidth]{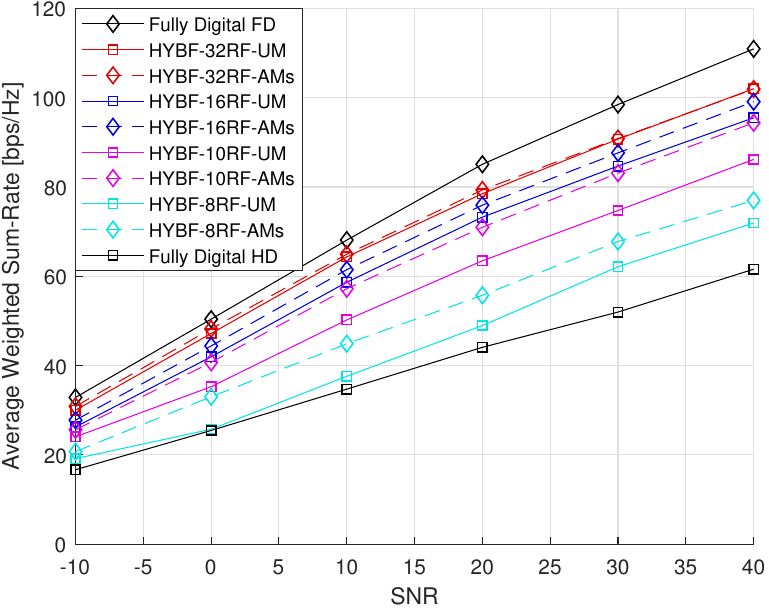}
   \caption{Average WSR as a function of the SNR with LDR noise $k_0= -80$~dB.}
    \label{SNR_32RF_08}
    \end{minipage}\hfill
    \begin{minipage}{0.49\textwidth}
      \centering
      \includegraphics[width=\textwidth]{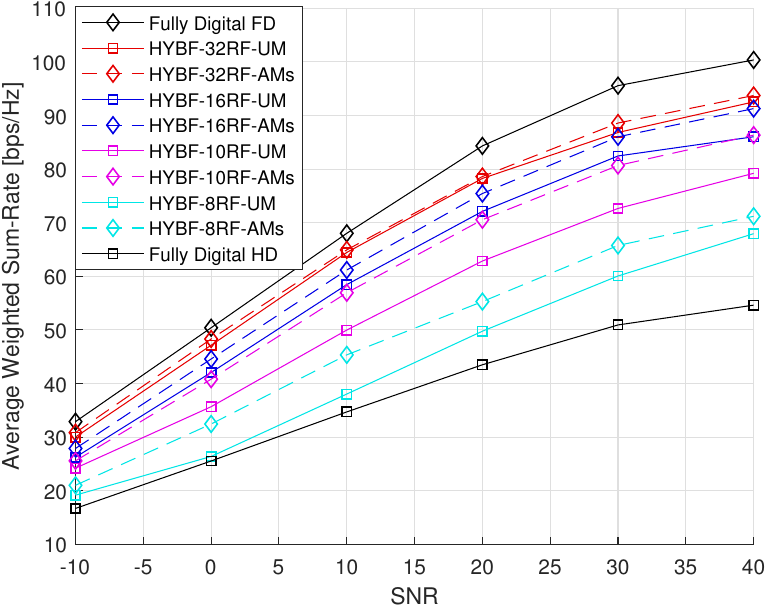}
   \caption{Average WSR as a function of the SNR with LDR noise $k_0= -60$~dB.}
    \label{SNR_32RF_06} 
 \end{minipage}\hfill
\end{figure*}
Figure \ref{SNR_32RF_08} shows the average WSR with a low LDR noise level $\kappa_0= -80~$dB with $32,16,10$ and $8$ RF chains as a function of the SNR. Both the proposed designs perform very close to the fully digital FD scheme with $32~$RF chains. HYBF-UM and HYBF-AMs outperform the fully digital HD scheme with only $8$ RF chains at high SNR and at any SNR level, respectively. It is evident the advantage of AMs, which add additional gain for all the SNR levels when the number RF chains at the FD BS is small. With a high number of RF chains, digital beamforming has enough amplitude manipulation liberty to manage the interference and adding AMs does not bring further improvement. Figure \ref{SNR_32RF_06} shows the average WSR achieved with a moderate LDR noise level $\kappa_0= -60~$dB. We can see that for a low SNR, the achieved average WSR results to be similar as reported in Figure \ref{SNR_32RF_08}. At high SNR, the LDR noise variance starts dominating, which leads to less achieved average WSR compared to the case of Figure \ref{SNR_32RF_08}. Figure \ref{SNR_32RF_04} shows the achieved WSR as a function of the SNR with a very large LDR noise variance of $\kappa_0= -40~$dB. By comparing the results reported in Figure \ref{SNR_32RF_04} and Figures \ref{SNR_32RF_08}-\ref{SNR_32RF_06}, we can see that the LDR noise variance dominates for most of the considered SNR range. For a very low SNR, the achieved WSR is similar as reported in Figures \ref{SNR_32RF_08}-\ref{SNR_32RF_06}. However, as the SNR increases, it does not map into higher WSR. It is clear that the maximum achievable WSR with $\kappa_0= -40~$dB saturates already at SNR$=20~$dB for both the HD and FD systems. Further improvement in the SNR does not dictate into higher WSR. When the LDR noise variance dominates, it acts as a ceiling to the effective received-signal-to-LDR-plus-thermal-noise-ratio (RSLTR). The transmit and receive LDR noise variance is proportional to the total transmit power per-antenna and received power per RF chain after the analog combining, respectively. When the LDR noise variance is large, the thermal noise variance has a negligible effect on the effective RSLTR. Consequently, a decrease in the thermal noise variance (increasing SNR) does not dictate a better WSR.
 
 \begin{figure}[t]
    \centering
      \includegraphics[width=0.49\textwidth]{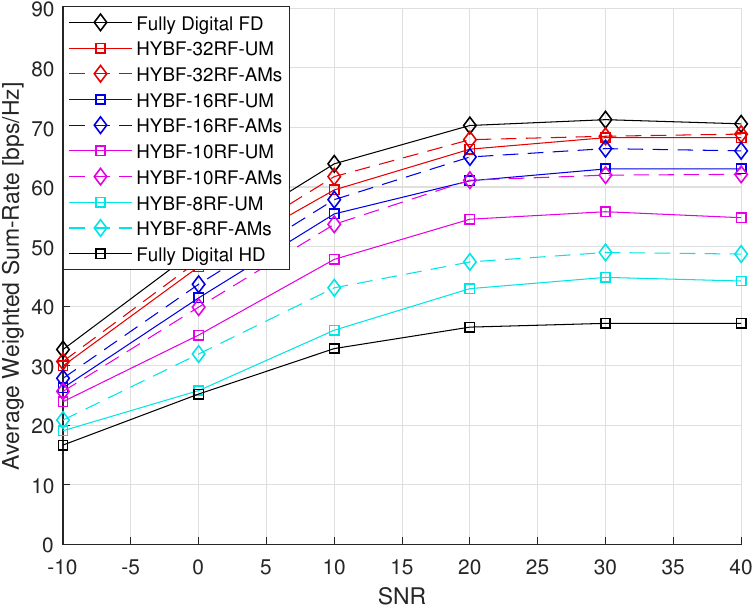}
   \caption{Average WSR as function of the SNR with LDR noise $k_0= -40$~dB.}
    \label{SNR_32RF_04}
\end{figure}

\begin{figure}[t]
   \centering
      \includegraphics[width=0.49\textwidth]{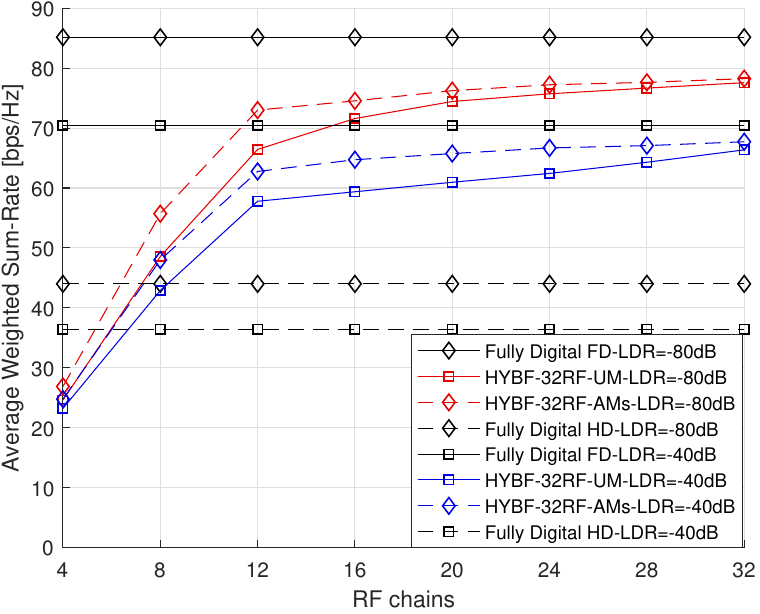}
   \caption{Average WSR as a function of the RF chains with LDR noise $k_0=-80$ dB and $k_0=-40$ dB at SNR$=20~$dB.}
    \label{RF_chains}
\end{figure}

Figure \ref{RF_chains} shows the achievable performance of HYBF-UM and HYBF-AMs as a function of the RF chains with SNR$=20~$dB, in comparison with the benchmark schemes, with very high and very small LDR noise levels. In particular, with very high LDR noise $k_k=-40~$dB and $8$ RF chains, HYBF-UM and HYBF-AMs perform close to the fully HD system, and an increase in the number of RF chains improves the performance, which tends towards the achieved WSR by a fully digital FD system with LDR noise level $k_k=-40~$dB. Similar behaviour can be observed for the case of low LDR noise $k_k=-80~$dB. Both the proposed schemes achieve higher WSR with the same number of RF chains in the latter case. We can also see that AMs add additional gain with a low number of RF chains, and as the number of RF chains increase, the gap in the achievable WSR with HYBF-AMs and HYBF-UM closes. In particular, with $32$ RF chains, the difference in the WSR with or without AMs becomes negligible.

From the results reported in Figures \ref{fig_100ant_0dB}-\ref{RF_chains}, we can conclude that the proposed HYBF schemes achieve significant performance improvement, in terms of average WSR, compared to a fully digital HD system. LDR noise plays a key role in determining the maximum achievable WSR for both the FD and HD systems. Figures \ref{fig_100ant_0dB}-\ref{fig_100ant_40dB} shows how an increase in the LDR noise variance degrades the average WSR at low and high SNR levels. Figures \ref{SNR_32RF_08}-\ref{SNR_32RF_06} shows that with a large to moderate dynamic range, the LDR noise degrades the performance only at very high SNR. Figure \ref{SNR_32RF_04} shows the achieved WSR as a function of a very large LDR noise variance. In that case, it is observed that the WSR saturates at SNR$=20~$dB and further improvement in the SNR does not dictate higher WSR. From Figure \ref{RF_chains}, it is clear how the number of RF chains at the mmWave FD BS affects the achievable WSR with different LDR noise levels and with or without the AMs.

\section{Conclusion} \label{conclusioni}
This paper has presented a novel HYBF design to maximize the WSR in a single-cell mmWave FD system with multi-antenna users and suffering from LDR noise. The beamformers were designed under the joint sum-power and the practical per-antenna power constraints. Simulation results showed that the multi-user mmWave FD systems can outperform the fully digital HD system with only a few RF chains. The advantage of having amplitude control at the analog processing stage is also investigated, and the benefit resulted to be evident with a small number of RF chains. Achievable average WSR with different levels of the LDR noise variance is also investigated, and the proposed HYBF designs outperformed the fully digital HD system at any LDR noise level.

\appendices
\section{Gradient Derivation} \label{grad_proof}
 The proof of Theorem~\ref{thm_grad} is based on the result derived in the following.

\begin{lemma} \label{lemma_derivative}
Let $\bm{Y} = \bm{A X B} + a\;\bm{A}\; \mbox{diag}\big(\bm{X} + \bm{Q}\big) \bm{B} + b\; \mbox{diag}\big(\bm{C X D} + \bm{E}\big) + \bm{F}$. 
The derivative of $\mbox{lndet}\big(\bm{Y}\big)$ with respect to $\bm{X}$ is given by 
\begin{equation}
    \begin{aligned}
     \frac{\partial \mbox{lndet}\bm{Y}}{\partial \bm{X}} = &\bm{A}^H \bm{Y}^{-H} \bm{B}^H + a\; \mbox{diag}\big(\bm{A}^H \bm{Y}^{-H} \bm{B}^H\big) \\& + b\; \bm{C}^H \mbox{diag}\big( \bm{Y}^{-H}\big) \bm{D}^H.
    \end{aligned} \label{lemma_result}
\end{equation}
\end{lemma}
\begin{proof}
By substituting  $\phi = \mbox{lndet}(\bm{Y})$, we can write
\begin{equation}
    \partial \phi = \bm{Y}^{-H}:d\bm{Y} = \mbox{Tr}\big(\bm{Y}^{-1} d\bm{Y}\big),
\end{equation}
where operator $:$ denotes the Frobenius inner product, i.e. $\bm{G}_{RF}:\bm{H} = \mbox{Tr}\big(\bm{G}_{RF}^H \bm{H}\big)$. Its derivative with respect to $\bm{X}$ can be written as
 \begin{equation}
\begin{aligned}
   \frac{\partial \phi}{\partial \bm{X}}& = \bm{Y}^{-H}: \big[\frac{d}{\partial \bm{X}}\big(\bm{A X B} + a\;\bm{A} \mbox{diag}\big(\bm{X}\big) \bm{B} \\& + b\; \mbox{diag}\big(\bm{C X D} + \bm{E}\big) + \bm{F}\big) \big)\big],
\end{aligned}
\end{equation}
where the last term results to be zero as independent from $\bm{X}$. Substituting the Forbenius product with the trace operator, using its cyclic shift and separating terms, yields
\begin{equation}
    \begin{aligned}
     \frac{\partial \phi}{\partial \bm{X}} = & \underbrace{\frac{\partial \;\mbox{Tr}\big( \bm{B} \bm{Y}^{-1} \bm{A X}\big)}{\partial \bm{X}}}_{\RomanNumeralCaps{1}} + \underbrace{a \frac{\partial\; \mbox{Tr}\big(\bm{B} \bm{Y}^{-1}  \bm{A} \mbox{diag}(\bm{X})\big)}{\partial \bm{X}}}_{\RomanNumeralCaps{2}} \\& + \underbrace{b \frac{\partial\; \mbox{Tr}\big(\bm{Y}^{-1}  \mbox{diag}(\bm{C X D})\big) }{\partial \bm{X}}}_{\RomanNumeralCaps{3}} + b \frac{\partial \;\mbox{Tr}\big(\bm{Y}^{-1} \mbox{diag}(\bm{E})\big)}{\partial \bm{X}},
    \end{aligned}
\end{equation}

where the last term being independent of $\bm{X}$ is also zero. To proof the aforementioned result, we proof the derivatives of  $\RomanNumeralCaps{1}, \RomanNumeralCaps{2}$ and $\RomanNumeralCaps{3}$ separately. Firstly, for $\RomanNumeralCaps{1}$, by using $:$ and doing some simple algebric manipulations leads to
\begin{equation}
    \frac{\partial \;\mbox{Tr}\big(\bm{B} \bm{Y}^{-1} \bm{A} \bm{X}\big)}{\partial \bm{X}} = \bm{A}^H \bm{Y}^{-H} \bm{B}^H:\partial\bm{X} = \bm{A}^H \bm{Y}^{-H} \bm{B}^H.
\end{equation}
To obtain the derivative of $\RomanNumeralCaps{2}$, we first define $\mbox{diag}\big(\bm{X}\big) = \bm{Z}$. The diagonal of $\bm{X}$ can be written as $\mbox{diag}\big(\bm{X}\big) = \bm{I} \circ \bm{X}$ 
where $\circ$ denotes the Hadamard product. By writing $\RomanNumeralCaps{2}$ with $:$ and expressing the diagonal term as a function of $\circ$, and using the 
commutative property of the Hadamard product leads to the following result
\begin{equation} \label{hadr}
    \begin{aligned}
     a\; \frac{\partial\; \mbox{Tr}\big(\bm{B} \bm{Y}^{-1}  \bm{A} \bm{Z}\big)}{\partial \bm{Z}} & =a\; \bm{A}^H \bm{Y}^{-H} \bm{B}^H: \partial\bm{Z}, \\ & =a\; \bm{A}^H \bm{Y}^{-H} \bm{B}^H: \bm{I} \circ \partial \bm{X}, \\& =a\; \bm{A}^H \bm{Y}^{-H} \bm{B}^H \circ \bm{I}:\partial \bm{X},\\&
     = a \;\mbox{diag}\big(\bm{A}^H \bm{Y}^{-H} \bm{B}^H\big).
    \end{aligned}
\end{equation}
To compute the derivative of $\RomanNumeralCaps{3}$, we first define $\mbox{diag}\big(\bm{C} \bm{X} \bm{D}\big) = \bm{W}$. By using a similar approach as in \eqref{hadr}, we get
\begin{equation}
    \begin{aligned}
     b\; \frac{\partial\; \mbox{Tr}\big(\bm{Y}^{-1}\bm{W}\big)}{\partial \bm{W}} & = b\;\bm{Y}^{-H}:\partial\bm{W}, \\& = b\bm{Y}^{-H}: \bm{I}\circ \bm{C} \partial \bm{X} \bm{D}, \\ &
    = b\;\bm{Y}^{-H} \circ \bm{I}: \bm{C} \partial\bm{X} \bm{D}, \\ &
    = b\; \mbox{diag}\big(\bm{Y}^{-H}\big) :\bm{C} \partial \bm{X} \bm{D},
    \\ & =  b\; \bm{C}^H \mbox{diag}\big( \bm{Y}^{-1}\big)^H \bm{D}^H.
    \end{aligned}
\end{equation}

Combining the result from each term concludes the proof for Lemma \eqref{lemma_result}. 
\end{proof}
To prove Theorem \ref{thm_grad}, note that the covariance matrices in \ref{covariance_matrices} has a special (Hermitian) structure, i.e., $\bm{B}= \bm{A}^H$ and $\bm{D}= \bm{C}^H$. Therefore, the result of Lemma \ref{lemma_derivative} for this particular case is given in the following.

\begin{lemma} \label{Lemma_specialCase}
 Let $\bm{Y} = \bm{A X B} + a\;\bm{A} \mbox{diag}\big(\bm{X} + \bm{Q}\big) \bm{B} + b\; \mbox{diag}\big(\bm{C X D} + \bm{E} \big) + \bm{F}$, where the size of matrices involved is such that the product is valid. Let $\bm{B} = \bm{A}^H$ and $\bm{D} = \bm{C}^H$ and the derivative of $\mbox{lndet}(\bm{Y})$ is given by
 \begin{equation}
    \begin{aligned}
     \frac{\partial \mbox{lndet}\bm{Y}}{\partial \bm{X}} = &\bm{A}^H \bm{Y}^{-H} \bm{A} + a\; \mbox{diag}\big(\bm{A}^H \bm{Y}^{-H} \bm{A}\big) \\& + b\; \bm{C}^H \mbox{diag}\big(\bm{Y}^{-H}\big) \bm{C}.
    \end{aligned} \label{lemma_result1}
\end{equation}
\end{lemma}
\begin{proof}
The result follows directly by relying on the result given in Lemma \ref{lemma_derivative} by substituting $\bm{B} = \bm{A}^H$ and $\bm{D} = \bm{C}^H$

\end{proof}

\begin{proof}\emph{Theorem \ref{thm_grad}}
To prove the gradients to linearize the WSR with respect to $\bm{T}_k$ and $\bm{Q}_j$, we proceed by simplifying the WSR as 
\begin{equation}
\begin{aligned}
  \mbox{WSR}   &= \sum_{k \in \mathcal{U}} w_k  \mbox{lndet} \Big(\bm{R}_{k} \Big) -  w_k \mbox{lndet} \Big(\bm{R}_{\overline{k} }\Big) \\&+ \quad \sum_{j \in \mathcal{D}}  w_j \mbox{lndet}\Big(\bm{R}_j\Big) -  w_j  \mbox{lndet}\Big( \bm{R}_{\overline{j}}\Big) .
\end{aligned}
\end{equation}

The $\mbox{WSR}_{\overline{k}}^{UL}$ and $\mbox{WSR}^{DL}$ should be linerized for $\bm{T}_k$ and $\mbox{WSR}_{\overline{j}}^{DL}$ and $\mbox{WSR}^{UL}$ for $\bm{Q}_j$. Note from \eqref{covariance_matrices} that $\bm{T}_k$ appears in $\mbox{WSR}_{\overline{k}}^{UL}$ and  $\mbox{WSR}^{DL}$ with the structure 
$\bm{Y} = \bm{A X} \bm{A}^H + a\;\bm{A}\; \mbox{diag}\Big(\bm{X} + \bm{Q}\Big) \bm{A}^H + b\; \mbox{diag}\Big(\bm{C X} \bm{C}^H + \bm{E} \Big) + \bm{F}$, where the scalars $a$ and $b$ are due to the LDR noise model, $\bm{A}$ and $\bm{C}$ are the interfering channels, $\bm{F}$ and $\bm{E}$ contain the noise contributions from other transmit covariance matrices but independent from $\bm{T}_k$.  The same structure holds also for the DL covariance matrices $\bm{Q}_j, \forall j \in \mathcal{D}$. By applying the result from Lemma \ref{Lemma_specialCase} with $\bm{Y} = \bm{R}_k$ or $\bm{Y} = \bm{R}_{\overline{k}}$ repetitively  $K-1$ time for linearizing  $\mbox{WSR}_{\overline{k}}$ with respect to $\bm{T}_k$ yield the gradient $\bm{A}_{k}$. Similarly, by considering $\bm{Y} = \bm{R}_j$ or $\bm{Y} = \bm{R}_{\overline{j}}$, $\forall j \in \mathcal{D}$ and applying the result from Lemma \ref{Lemma_specialCase} yield the gradient $\bm{B}_{k}$.

The same reasoning holds also for $\bm{Q}_j$, which leads to the gradients $\hat{\bm{C}}_j$ and $\bm{D}_j$ by applying the result provided in Lemma \ref{Lemma_specialCase} for $\mbox{WSR}_{\overline{j}}^{DL}$ $J-1$ times and for $\mbox{WSR}^{UL}$ $K$ times, respectively, $\forall j \in \mathcal{D}$.
\end{proof}

\section{Proof of Theorem \textbf{\ref{optimal_hybrid_precoding}}} \label{hybrid_appendix}
The dominant generalized eigenvector solution maximizes the reformulated concave WSR maximization problem 

\begin{equation}
\begin{aligned}
\mbox{WSR}& =\sum_{k \in \mathcal{U}} w_k\mbox{lndet}\Big(\bm{I} + \bm{U}_k^H  \bm{H}_k^H \bm{F}_{RF}  \bm{R}_{\overline{k}}^{-1}  \bm{F}_{RF}^H \bm{H}_k \bm{U}_k  \Big) 
\\& -  \mbox{Tr}\Big(\bm{U}_k^H \Big(\hat{\bm{A}}_k + \hat{\bm{B}}_k + l_k \bm{I}+ \bm{\Psi}_k \Big) \bm{U}_k  \Big) \\
 & + \sum_{j\in \mathcal{D}} w_j \mbox{lndet}\Big( \bm{I} + \bm{V}_j^H \bm{G}_{RF}^H \bm{H}_j^H \bm{R}_{\overline{j}}^{-1} \bm{H}_j \bm{G}_{RF} \bm{V}_j\Big) \\& - \mbox{Tr}\Big(\bm{V}_j^H \bm{G}_{RF}^H \Big(\hat{\bm{C}_j} + \hat{\bm{D}_j} +  l_0 \bm{I}+  \bm{\Psi}_0 \Big) \bm{G}_{RF} \bm{V}_j  \Big) \Big).
\end{aligned} \label{eq_restate}
\end{equation}

To prove Theorem \ref{optimal_hybrid_precoding} for solving \eqref{eq_restate}, we first consider the UL digital beamforming solution by keeping the analog beamformer and the digital DL beamformers fixed. We proceed by considering user $k \in \mathcal{U}$ for which we wish to compute the WSR maximizing digital UL beamformer. The same proof will be valid $\forall k \in \mathcal{U}$. The proof relies on simplifying 
\begin{equation}
\begin{aligned}
\underset{\bm{U}_k}{\text{max.}}\quad  & w_k\mbox{lndet}\Big(\bm{I} + \bm{U}_k^H  \bm{H}_k^H \bm{F}_{RF}  \bm{R}_{\overline{k}}^{-1}  \bm{F}_{RF}^H  \bm{H}_k \bm{U}_k  \Big)  \\&-  \mbox{Tr}\Big(\bm{U}_k^H \Big(\hat{\bm{A}}_k + \hat{\bm{B}}_k + l_k \bm{I}+ \bm{\Psi}_k \Big) \bm{U}_k  \Big)
\end{aligned} \label{UL_proof}
\end{equation}
until the Hadamard's inequality applies as in Proposition 1\cite{kim2011optimal} or Theorem 1 \cite{hoang2008noncooperative}. The Cholesky
decomposition of the matrix $\Big(\hat{\bm{A}}_k + \hat{\bm{B}}_k + l_k + \bm{\Psi}_k)$ is given as $\bm{L}_k \bm{L}_k^H$ where $\bm{L}_k $ is the lower triangular Cholesky factor. By defining $\tilde{\bm{U}_k} = \bm{L}_k^H \bm{U}_k$, \eqref{UL_proof} reduces to 
\begin{equation}
\begin{aligned}
\underset{\bm{U}_k}{\text{max.}} \;  w_k \mbox{lndet}\Big( \bm{I} & + \tilde{\bm{U}_k}^H \bm{L}_k^{-1} \bm{H}_k^H \bm{F}_{RF}  \bm{R}_{\overline{k}}^{-1}  \bm{F}_{RF}^H \bm{H}_k \\& \bm{L}_k^{-H} \tilde{\bm{U}_k} \Big)  -  \mbox{Tr}\Big(\tilde{\bm{U}_k}^H \tilde{\bm{U}_k}  \Big).
\end{aligned} \label{UL_proof_1}
\end{equation}
Let $\bm{E}_k \bm{D}_k \bm{E}_k^H $ be the eigen-decomposition of $\bm{L}_k^{-1} \bm{H}_k^H \bm{R}_{\overline{k}}^{-1} \bm{H}_k \bm{L}_k^{-H}$, where $\bm{E}_k$ and $\bm{D}_k$ are the unitary and diagonal matrices, respectively. Let $\bm{O}_k = \bm{E}_k^H \tilde{\bm{U}_k} \tilde{\bm{U}_k}^H \bm{E}_k$ and \eqref{UL_proof_1} can be expressed as 
\begin{equation}
\begin{aligned}
\underset{\bm{O}_k}{\text{max.}}\quad  & w_k \mbox{lndet}\Big( \bm{I} + \bm{O}_k \bm{D}_k  \Big) -  \mbox{Tr} \Big(\bm{O}_k \Big).
\end{aligned} \label{UL_proof_2}
\end{equation}
By Hadamard’s inequality [Page 233 \cite{cover1999elements}]
, it can be seen that the optimal $\bm{O}_k$ must be diagonal. Therefore, $\bm{U}_k = \bm{L}_k^{-H}\bm{E}_k \bm{O}_k^{\frac{1}{2}} $ and thereby
\begin{equation}
\begin{aligned}
          \bm{H}_k^H \bm{F}_{RF}  \bm{R}_{\overline{k}}^{-1}  \bm{F}_{RF}^H & \bm{H}_k  \bm{U}_k =  \bm{L}_k \bm{L}_k^{H} \bm{L}_k^{-H} \bm{E}_k  \bm{O}_k^{\frac{1}{2}} \bm{D}_k\\& =   \Big(\hat{\bm{A}}_k + \hat{\bm{B}}_k + l_k + \bm{\Psi}_k \Big) \bm{U}_k \bm{D}_k, \label{solution_gev}
\end{aligned}
\end{equation}
from which we select $u_k$ dominant eigenvectors, which concludes the proof for the UL beamformer for user $k \in \mathcal{U}$. For the digital DL beamformers the proof follow similarly by considering the following optimization problem $\forall j$

\begin{equation}
\begin{aligned}
\underset{\bm{V}_j}{\text{max.}}&
 w_j \mbox{lndet}\Big( \bm{I} + \bm{V}_j^H \bm{G}_{RF}^H \bm{H}_j^H \bm{R}_{\overline{j}}^{-1} \bm{H}_j \bm{G}_{RF} \bm{V}_j\Big) \\& - \mbox{Tr}\Big(\bm{V}_j^H \bm{G}_{RF}^H \Big(\hat{\bm{C}_j} + \hat{\bm{D}_j} +  l_0 +  \bm{\Psi}_0 \Big) \bm{G}_{RF} \bm{V}_j \Big).
\end{aligned}
\end{equation}
and simplifying it until the Hadamard's inequality applies to yield a similar result as expressed in \eqref{solution_gev}.

The proof for analog beamformer $\bm{G}_{RF}$ does not apply directly as the KKT condition have the form $\bm{A}_1 \bm{G}_{RF} \bm{A}_2 = \bm{B}_1 \bm{G}_{RF} \bm{B}_2$, which are not resolvable. To solve it for the analog beamformer $\bm{G}_{RF}$, we apply the result $\mbox{vec}(\bm{A} \bm{X} \bm{B}) = \bm{B}^T \otimes \bm{A}\mbox{vec}(\bm{X})$ \cite{magnus2019matrix}, which allows to rewrite \eqref{grad_analog_precoder} as 
\begin{equation}
\begin{aligned}
    \sum_{j \in \mathcal{D}}& w_j \Big( \Big(\bm{V}_j \bm{V}_j^H \Big(\bm{I} + \bm{V}_j \bm{V}_j^H \bm{G}_{RF}^H \bm{H}_j^H \bm{R}_{\overline{j}}^{-1} \bm{H}_j \bm{G}_{RF} \Big)^{-1}\Big)^T \otimes \\
      & \bm{H}_j^H \bm{R}_{\overline{j}}^{-1}  \bm{H}_j \Big) \mbox{vec}\Big(\bm{G}_{RF}\Big) - \sum_{j \in \mathcal{D}} \Big(\Big(\bm{V}_j \bm{V}_j^H \Big)^T \otimes \Big(  \hat{\bm{C}_j}\\&  + \hat{\bm{D}_j}
     + \bm{\Psi}_0  + l_0 \bm{I}\Big)\Big) \mbox{vec}\Big(\bm{G}_{RF}\Big)=0.
\end{aligned} \label{restatement_KKT_analog}
\end{equation}

The WSR maximizing analog beamformer can alternatively be derived as follows (which allows the proof for the digital beamformers to be applicable directly). First we apply a noise whitening procedure using the noise plus interference covariance matrix $\mathbf{R}_{\overline{j}}^{1/2}$on the received signal. Further, we can rewrite the whitened signal as follows
\begin{equation}
\begin{aligned}
         \widetilde{\bm{y}}_j =   &\left(\left(\bm{s}_{j_d}^T\bm{V}_j^T\right)\otimes \mathbf{R}_{\overline{j}}^{-1/2}\bm{H}_j\right)\mbox{vec}(\bm{G}_{RF}) + \widetilde{\mathbf{n}}_j ,
\end{aligned}
    \label{Rx_side_1}
\end{equation}

where $\widetilde{\mathbf{y}}_j=\mathbf{R}_{\overline{j}}^{-1/2}\mathbf{y}_j$ and $\widetilde{\mathbf{n}}_j$ represents the whitened noise plus interference signal. We can write the resulting WSR optimization problem, after the approximation to concave form and some algebraic manipulations on the linearized term, as

\begin{equation} 
\begin{aligned}
\underset{\bm{G}_{RF}}{\text{max}} \quad   \sum_{j \in \mathcal{D}} & w_j  
 \mbox{lndet}\Big(\bm{I} +     \mbox{vec}\Big(\bm{G}_{RF}\Big) ^H\Big(\Big(\bm{V}_j \bm{V}_j^H \Big)^T \otimes \bm{H}_j^H \bm{R}_{\overline{j}}^{-1} \\& \bm{H}_j \Big) \mbox{vec}\Big(\bm{G}_{RF}\Big) \Big) - \mbox{Tr} \Big( \mbox{vec}\Big(\bm{G}_{RF}\Big)^H\Big( \bm{V}_j\bm{V}_j^H \otimes  \\ &   \Big(\hat{\bm{C}}_{j}  +  \hat{\bm{D}}_{j}\Big) + \bm{\Psi}_0  + l_0 \bm{I} \Big) \mbox{vec}\Big(\bm{G}_{RF}  \Big) \Big).  
\end{aligned} \label{WSR_convex_alt}
\end{equation}

Taking the derivative of \eqref{WSR_convex_alt} for the conjugate of $\mathbf{G}_{RF}$ leads to the same generalized eigenvector solution as in \eqref{hybrid_AP}. Note that this alternative representation has the same form as  \eqref{UL_proof}, which is resolvable for the vectorized version of the analog beamformer $\bm{G}_{RF}$. Therefore, the proof for the UL and DL digital beamformers can now be applied directly on the vectorized analog beamformer $\mbox{vec}(\bm{G}_{RF})$, which is summed over all the DL users served by the mmWave FD BS.

\section*{Acknowledgment}

EURECOMs research is also partially supported
by its industrial members: ORANGE, BMW, Symantec,
SAP, Monaco Telecom, iABG, and by the projects MASS-START
(French FUI), DUPLEX (French ANR), SPOTLIGHT (EU ITN)
and the Qualcomm Fab5G project.

\bibliographystyle{IEEEtran}
\bibliography{main.bib}

\end{document}